\documentclass{article}

\usepackage{amsmath}
\usepackage{cite}

\usepackage{amssymb,color}
\usepackage{appendix}
\usepackage{amsthm}
\usepackage{enumerate}
\usepackage{enumitem}
\usepackage{hyperref}
\usepackage{authblk}
\usepackage{mathtools} 
\mathtoolsset{showonlyrefs=true}  
\usepackage{bm}
\usepackage{bbm}
\usepackage{setspace}
\usepackage[margin=1in]{geometry}
\usepackage{graphicx}

\usepackage{tabularx, booktabs, makecell, caption, rotating, multirow} 
\usepackage{siunitx} 

\numberwithin{equation}{section}

\theoremstyle{plain}

\newtheorem{theorem}{Theorem}[section]
\newtheorem{corollary}[theorem]{Corollary}
\newtheorem{definition}[theorem]{Definition} 
\newtheorem{assumption}{Assumption}
\newtheorem{proposition}[theorem]{Proposition} 
\newtheorem{lemma}[theorem]{Lemma} 
\newtheorem{example}[theorem]{Example}

\begin{document}

\author{
	{Jiuk Jang\textsuperscript{a}
		\quad Jaehyun Kim\textsuperscript{b}\thanks{jaehyunkim@cuhk.edu.hk, jaehyunkim2026@gmail.com}
		\quad Hyungbin Park\textsuperscript{a,c}
		\quad Jonghwa Park\textsuperscript{d}
	}\\ $ $ \\

	{\textsuperscript{a}{\small Department of Mathematical Sciences, Seoul National University, 1, Gwanak-ro, Gwanak-gu, Seoul, South Korea;}\\
	\textsuperscript{b}{\small Department of Statistics and Data Science, The Chinese University of Hong Kong, Shatin, N.T., Hong Kong;}\\ 		
		 \textsuperscript{c}{\small Research Institute of Mathematics, Seoul National University, 1, Gwanak-ro, Gwanak-gu, Seoul, South Korea;} \\
		\textsuperscript{d}{\small Department of Mathematical Sciences, Carnegie Mellon University,  Pittsburgh, Pennsylvania, USA  }}
}


\title{Hedging short-maturity Asian options in local volatility models}

	\date{\today}   
	
\maketitle

\begin{abstract}
This paper discusses the short-maturity behavior of Asian option prices and hedging portfolios.	
We consider the risk-neutral valuation and the  delta value of the Asian option  having a H\"older continuous payoff function  in a local volatility model. 
The main idea of this analysis is that the local volatility model can be approximated by a Gaussian process at short maturity.
By combining this approximation argument with Malliavin calculus, 
we derive short-maturity asymptotics for Asian option prices and deltas, and express them in terms of the local volatility function and the initial stock price.
In addition, we show that the convergence rate of the approximation is determined by  the H\"older exponent of the payoff function. Numerical experiments on concrete examples validate the effectiveness of the proposed method.\\


\end{abstract}

\section{Introduction}

This paper investigates an arithmetic average Asian option in continuous time with a terminal payoff of the form
\begin{equation}
\Phi\left(\frac{1}{T}\int_0^T S_t\,dt\right)
\end{equation}
where   $\Phi:\mathbb{R}\rightarrow\mathbb{R}$ is a given payoff function, $T>0$
denotes the maturity, and $(S_t)_{t\ge 0}$
represents the underlying asset price process.
We refer to this contract simply as the Asian option.
Because of its average property, the Asian option is less exposed to a sudden plummet in stock prices just before maturity. This property makes it particularly appealing for hedging purposes and widely used by traders and financial institutions. For a detailed discussion of the role of Asian options in financial markets, see \cite{WilmottPaul2006PWoq}.
Despite its popularity in real markets, the Asian option is mathematically challenging to price and hedge. Even when the underlying asset $(S_t)_{t\ge 0}$ follows the classical Black--Scholes model, no simple closed-form expression for the density of the random variable $\frac{1}{T}\int_0^T S_t\,dt$ is available. The path-dependent nature of the payoff introduces significant analytical and computational difficulties, making asymptotic and approximation methods particularly valuable.

The main contribution of this paper is an asymptotic analysis of the initial-value sensitivity of short-maturity Asian options, commonly referred to as the delta in the finance literature.
To the best of our knowledge, rigorous analytical results for short-maturity deltas have not been established even in the Black--Scholes model. The payoff function $\Phi$ is assumed to be H\"older continuous, and the price process $(S_t)_{t\ge 0}$ follows a local volatility model; the detailed assumptions are provided in Section \ref{sec:price_main}.
Define the Asian volatility by
\begin{equation}
\sigma_A(T):=\sqrt{\frac{1}{T^3}\int_0^T \sigma^2(t,S_0)(T-t)^2\,dt}\,,
\end{equation}
where $\sigma(\cdot,\cdot)$ denotes the local volatility function.
We show that, for small maturity $T>0$,
the delta admits an asymptotic representation in terms of the Asian volatility.
More precisely, the  delta $\Delta_A(T)$ satisfies
\begin{align}
\Delta_A(T)=\mathbb{E}^\mathbb{Q}\left[\frac{\Phi(S_0+S_0\sigma_A(T)\sqrt{T}Z)}{S_0\sigma_A(T)\sqrt{T}}Z\right]+\mathcal{O}(T^{\gamma-\frac{1}{2}})\,,
\end{align}
where $Z$ is a standard normal random variable and $\gamma$ denotes the H\"older exponent of the payoff function  $\Phi.$

As a secondary contribution, short-maturity Asian option prices are examined. Let $P_A(T)$ denote the Asian option price with maturity $T$. The asymptotic behavior as $T\to 0$ is characterized by
\begin{equation}
P_A(T)=\mathbb{E}^\mathbb{Q}[\Phi(S_0+S_0\sigma_A(T)\sqrt{T}Z)]+\mathcal{O}(T^{\gamma})  \,. 
\end{equation}
The novelty here is that H\"older-continuous payoffs are treated in a unified and rigorous way, a setting that has not been addressed in the existing literature. At the same time, finer short-maturity expansions are already available for standard call and put payoffs; for this reason, the pricing result is presented as a minor contribution relative to the main focus on deltas.

To obtain our results, we combine several well-established mathematical techniques with an approximation scheme.
The core idea is an $L^p$-approximation of the underlying stock price process $(S_t)_{0\le t\le T}$ by a suitably constructed Gaussian process $(\hat{X}_t)_{0\le t\le T}.$ 
Our approach builds upon the methods of \cite{PirjolDan2016SMAO, PirjolDan2019SMAO, PirjolDan2019SMFS}, where a similar approximation was employed to derive short-maturity asymptotics for at-the-money (ATM) Asian call and put option prices.
Following this idea, our analysis focuses on the Gaussian random variable  $\frac{1}{T}\int_0^T S_t\,dt\,$ having a sophisticated density to the Gaussian random variable $\frac{1}{T}\int_0^T \hat{X}_t\,dt$
instead of the original average
$\frac{1}{T}\int_0^T S_t\,dt$, whose density is analytically intractable. This forms the core strategy underlying our approximation framework.
In addition, we employ Malliavin calculus to analyze the Asian option delta.
This technique has been successfully used for sensitivity analysis of Asian call and put options in \cite{benhamou2000application, PirjolDan2018SOAO}, and we adapt their approach to obtain an explicit representation for the delta.
Furthermore, the large deviation principle is applied to investigate both 
OTM and ITM Asian call and put options, following the line of research initiated in \cite{PirjolDan2016SMAO, PirjolDan2019SMAO}.

In recent years, short-maturity Asian options have been studied by many researchers.
Under a local volatility model, the asymptotic behavior of Asian option prices has been investigated in \cite{PirjolDan2016SMAO, PirjolDan2019SMFS}.
In \cite{PirjolDan2019SMAO}, asymptotic analysis was further developed under the constant elasticity of variance (CEV) model.
These studies employ the large deviation principle and derive analytical expressions for the rate function of the distribution of 
$\frac{1}{T}\int_0^T S_t\,dt$ 
to construct asymptotic approximations.
Sensitivity analysis was later performed in \cite{PirjolDan2018SOAO} as a follow-up to \cite{PirjolDan2016SMAO}, focusing on the delta under the Black--Scholes framework using the previously obtained price approximation.
More recently, \cite{shoshi2025some} analyzed short-maturity Asian options via large deviations theory within a constant volatility model that incorporates a jump component in addition to the drift and diffusion terms.
Within this framework, they derived asymptotic estimates for out-of-the-money, in-the-money and at-the-money Asian call and put options.
Their results rely heavily on the works of \cite{linetsky2004spectral} and \cite{sengupta2014pricing}.

The remainder of this paper is organized as follows.
Section \ref{sec:price_main} analyzes the short-maturity behavior of Asian option prices.
Section \ref{sec:Short maturity asymptotic for a sensitivity} establishes preliminary estimates for the short-maturity deltas of Asian and European options.
Section \ref{sec:delta_main} presents the main results of the paper, providing short-maturity asymptotics for the deltas of Asian and European options.
Section \ref{sec:numerical} reports numerical experiments that support and illustrate our theoretical estimates.
Section \ref{sec:Special case} applies the large deviation principle to analyze Asian call and put options.
Finally, Section \ref{sec:conclusion} concludes the paper.
Detailed proofs of necessary lemmas are collected in the appendices.

\section{Short-maturity option prices}
\label{sec:price_main}

We analyze the short-maturity asymptotic behavior of Asian options under local volatility
models. Assume that  the stock price process $(S_t)_{t\ge0}$ follows a local volatility model,
\begin{align}\label{eq:S_t}
dS_t=(r-q)S_t\,dt+\sigma(t,S_t)S_t\,dW_t\,, \quad S_0>0\,,
\end{align}
under risk-neutral measure $\mathbb{Q}$, where $r$ is the short rate, $q$ is the dividend rate, and $(W_t)_{t\ge0}$ is a $\mathbb{Q}$-Brownian motion. 
Under Assumption \ref{classical assumption}, the mapping $x\mapsto \sigma(t,x)x$ is Lipschitz continuous uniformly in 
$t$; hence, the SDE \eqref{eq:S_t} admits a unique strong solution \cite[Theorem 3.3.1]{zhang2017backward}.
Short-maturity Asian options under assumptions similar to those stated below have been investigated in \cite{shoshi2025some}.

\begin{assumption}
   \label{classical assumption}
	Let us consider the following conditions on the diffusion function.
	\begin{enumerate}
    \renewcommand{\labelenumi}{(\roman{enumi})}
			\renewcommand{\theenumi}{\roman{enumi}} 
		\item The function $\sigma:[0,T]\times\mathbb{R} \to\mathbb{R}$ is measurable, and there exist positive constants $\underline{\sigma}$ and $\overline{\sigma}$ such that $\underline{\sigma}\le\sigma(t,x) \le\overline{\sigma}$ for all $(t,x)\in [0,T]\times\mathbb{R}$.
		\item For each $t\in [0,T],$ the function $\sigma(t,\cdot)$ is twice continuously differentiable in $\mathbb{R}\,.$
		\item Define $\nu(t,x):=\frac{\partial[\sigma(t,x)x]}{\partial x}$ and $\rho(t,x):=\frac{\partial^2[\sigma(t,x)x]}{\partial x^2}\,.$ Then, the functions $\sigma(t,\cdot),$ $\sigma(t,\cdot)\cdot,$ $\nu(t,\cdot),$ $\rho(t,\cdot)$ are Lipschitz continuous uniformly in $t$. More precisely, there is a constant $\alpha>0$ such that for any $x,y\in\mathbb{R},$
		\begin{align}
		&\underset{0\le t\le T}{\sup}\vert \sigma(t,x)-\sigma(t,y) \vert \le \alpha \vert x-y \vert, \ \underset{0\le t\le T}{\sup}\vert \sigma(t,x)x-\sigma(t,y)y \vert \le \alpha \vert x-y \vert, \\
		&\underset{0\le t\le T}{\sup}\vert \nu(t,x)-\nu(t,y) \vert \le \alpha \vert x-y \vert, \ \underset{0\le t\le T}{\sup}\vert \rho(t,x)-\rho(t,y) \vert \le \alpha \vert x-y \vert\,.
		\end{align}
	\end{enumerate}
\end{assumption}
\noindent This assumption includes deterministic volatility functions, that is, positive continuous functions $\sigma : [0, T] \to \mathbb{R}$ that are independent of $x$. 

Although the assumptions may appear restrictive, our results provide a meaningful contribution in two respects. First, the delta asymptotics are novel even under this restrictive setting. To the best of our knowledge, there are no existing analytical results with mathematically rigorous proofs for short-maturity deltas, even in the Black--Scholes model. Second, a suitable truncation procedure allows a broad class of local volatility models used in practice to be brought within the scope of the assumptions. This truncation method, together with supporting numerical results, is presented in Section \ref{sec:LV}.

In this paper, we consider H\"older continuous payoff functions $\Phi$.  Under Assumption \ref{Holder payoff assumption}, the constants 
$\beta$ and $\gamma$ are always understood to be those associated with the payoff function 
$\Phi$  throughout this paper.

A specific class of H\"older continuous payoff functions has been studied in \cite{bernard2016power,denis2009leland,liu2012implied,liu2024power} in the context of European option pricing.

\begin{assumption}\label{Holder payoff assumption}
	The payoff function $\Phi:\mathbb{R}\rightarrow{\mathbb{R}}$ is $\gamma$-H\"older continuous with coefficient $\beta>0$. More precisely, for any $x,y\in\mathbb{R}$, $\vert\Phi(x)-\Phi(y)\vert \le \beta\vert x-y\vert^{\gamma}$ with $0<\gamma\le 1\,.$
\end{assumption}

Under the risk-neutral measure $\mathbb{Q}$, 
the arbitrage-free values of the Asian and European options are 
$$ P_{A}(T):=e^{-rT}\,\mathbb{E}^\mathbb{Q}\left[\Phi\left(\frac{1}{T}\int_0^T S_t\,dt\right)\right]\,,\quad   P_{E}(T):=e^{-rT}\,\mathbb{E}^\mathbb{Q}[\Phi(S_T)]\,,$$
where $T$ is the maturity. 
Our objective is to find an asymptotic formula for the Asian option price up to $\mathcal{O}(T)$; a formula for its European counterpart is also presented. 
We introduce two notions of volatilities called the {\em Asian volatility} and the {\em European volatility}.
\begin{definition} We define the Asian volatility $\sigma_A(T)$ and the European volatility $\sigma_E(T)$ as
	\begin{equation}\label{def:asian european volatility}
	\sigma_{A}(T):=\sqrt{\frac{1}{T^3}\int_0^T\sigma^2(t,S_0)(T-t)^2\,dt}\,, \; \sigma_{E}(T):=\sqrt{\frac{1}{T}\int_0^T\sigma^2(t,S_0)\,dt}\,.
	\end{equation}
\end{definition}

The following theorem is one of the main results of this paper. It characterizes the short-time asymptotic behavior of Asian and European option prices.
Related studies on short-maturity asymptotics of Asian option prices can be found in \cite[Theorem 2]{PirjolDan2016SMAO}, \cite[Theorem 5]{PirjolDan2019SMAO}, and \cite[Theorems 4.5 and 4.10]{shoshi2025some}.
The proof of the theorem is provided in Appendix~\ref{sec:Short maturity limit of an option price}.

\begin{theorem}\label{thm:option price holder}
	Under Assumptions \ref{classical assumption} and \ref{Holder payoff assumption}, as $T\rightarrow{0}$, we have
	\begin{equation}
	P_A(T)=\mathbb{E}^\mathbb{Q}[\Phi(S_0+S_0\sigma_A(T)\sqrt{T}Z)]+\mathcal{O}(T^{\gamma})\,,\quad
	P_E(T)=\mathbb{E}^\mathbb{Q}[\Phi(S_0+S_0\sigma_E(T)\sqrt{T}Z)]+\mathcal{O}(T^{\gamma})\,,
	\end{equation}
	where $Z$ is a standard normal random variable. 
\end{theorem}

\begin{example}\label{ex: gamma price holder}
This example demonstrates that the asymptotic convergence order $\mathcal{O}(T^{\gamma})$ in Theorem \ref{thm:option price holder} cannot be further improved.
Given any $K>0$ and $0<\gamma \le 1\,,$ define the payoff function $\Phi$ by
	\begin{equation}
	\Phi(x)=(x-K)_{+}^{\gamma}\,.
	\end{equation}
	If $S_0=K,$ then 
	\begin{equation}
	P_A(T)=\frac{1}{2}(S_0\sigma_A(T))^{\gamma}M(\gamma)T^{\frac{\gamma}{2}}+\mathcal{O}(T^{\gamma})\,,
	\end{equation}
	where $M(\gamma):=\mathbb{E}^\mathbb{Q}[\vert Z\vert^{\gamma}]$ for a standard normal variable $Z.$ If we replace $\sigma_A(T)$ with $\sigma_E(T)\,,$ we obtain the asymptotic result for the European option price $P_E(T)\,.$
\end{example}

	In Theorem \ref{thm:option price holder}, if $\Phi(S_0)\neq 0$, then the expectations
	$$\mathbb{E}^\mathbb{Q}[\Phi(S_0+S_0\sigma_A(T)\sqrt{T}Z)]\,\textnormal{ and }\;\mathbb{E}^\mathbb{Q}[\Phi(S_0+S_0\sigma_E(T)\sqrt{T}Z)]$$
	are of the order greater than $\mathcal{O}(T^\gamma).$
	If $\Phi(S_0)=0$, then the convergence order depends on the function $\Phi.$
	For the ATM (i.e., $S_0=K$) call and put options with $\Phi(x)=(x-K)_+^{\gamma}$ and $\Phi(x)=(K-x)_+^{\gamma}$, we have 
	$$\mathbb{E}^\mathbb{Q}[\Phi(S_0+S_0\sigma_A(T)\sqrt{T}Z)]=(S_0\sigma_A(T))^{\gamma}T^{\frac{\gamma}{2}}\mathbb{E}[Z^{\gamma}\mathbbm{1}_{\{Z\geq 0\}}]\,,$$
	$$\mathbb{E}^\mathbb{Q}[\Phi(S_0+S_0\sigma_E(T)\sqrt{T}Z)]=(S_0\sigma_E(T))^{\gamma}T^{\frac{\gamma}{2}}\mathbb{E}[Z^{\gamma}\mathbbm{1}_{\{Z\geq 0\}}]\,,$$
	which means that these are   of the order greater than $\mathcal{O}(T^\gamma).$

If $\Phi$ is Lipschitz (that is, $\gamma=1$), then we have the following corollary.
It is directly obtained by Theorem \ref{thm:option price holder} incorporating with 	$\mathbb{E}^\mathbb{Q} [\Phi (S_0+S_0\sigma_A(T)\sqrt{T}Z ) ]=\Phi(S_0)+\mathcal{O}(\sqrt{T})$ and $\mathbb{E}^\mathbb{Q} [\Phi (S_0+S_0\sigma_E(T)\sqrt{T}Z ) ]=\Phi(S_0)+\mathcal{O}(\sqrt{T})\,.$

\begin{corollary}\label{cor:option price limit same} Under Assumptions \ref{classical assumption}
	and \ref{Holder payoff assumption}, if $\gamma=1,$  
	the prices of both the Asian option and the European option share the same limit $\Phi(S_0)$ as $T\rightarrow{0}$ with the convergence order $\mathcal{O}(\sqrt{T}).$ More precisely,
	\begin{equation} P_{A}(T)=\Phi(S_0)+\mathcal{O}(\sqrt{T})\,, \quad P_{E}(T)=\Phi(S_0)+\mathcal{O}(\sqrt{T})\,.\end{equation}
\end{corollary}

We obtain asymptotic results for the call and put options as an example of Theorem \ref{thm:option price holder}. This generalizes the result in Theorem 6 of \cite{PirjolDan2016SMAO} for the ATM case.
\begin{example}\label{ex:call,put price}
	Let $P_A^{\textnormal{call}}$ and $P_A^{\textnormal{put}}$
	be the Asian call and put prices with  the strike $K,$ i.e., the payoff functions are   $\Phi(x)=(x-K)_+$ and $\Phi(x)=(K-x)_+,$ respectively.
	Then,  
	\begin{align}
	&P_A^{\textnormal{call}}(T)=\begin{cases}
	0+\mathcal{O}(\sqrt{T}), & \mbox{if  }S_0<K\,, \\
	\frac{S_0\sigma_A(T)}{\sqrt{2\pi}}\sqrt{T}+\mathcal{O}(T), & \mbox{if  }S_0=K \,,\\
	S_0-K+\mathcal{O}(\sqrt{T}), & \mbox{if  }S_0>K\,,
	\end{cases}\\
	&P_A^{\textnormal{put}}(T)=\begin{cases}
	K-S_0+\mathcal{O}(\sqrt{T}), & \mbox{if  }S_0<K \,,\\
	\frac{S_0\sigma_A(T)}{\sqrt{2\pi}}\sqrt{T}+\mathcal{O}(T), & \mbox{if  }S_0=K\,, \\
	0+\mathcal{O}(\sqrt{T}), & \mbox{if  }S_0>K\,.
	\end{cases}
	\end{align}
	The prices of the European call and put option are obtained by replacing $\sigma_A(T)$ in the above-mentioned expressions with $\sigma_E(T)$.
\end{example}

	The convergence rate of the expectations
	$$\mathbb{E}^\mathbb{Q}\left[\Phi\left(S_0+S_0\sigma_A(T)\sqrt{T} Z\right)\right]\;\text{ and }\; \mathbb{E}^\mathbb{Q}\left[\Phi\left(S_0+S_0\sigma_E(T)\sqrt{T}Z\right)\right]$$
	are important   because
	if these are of the order smaller than or equal to $\mathcal{O}(T^\gamma),$ then Theorem \ref{thm:option price holder} will not be informative.
	If $\Phi(S_0)\neq 0$, it is easy to show that these are of the order greater than  $\mathcal{O}(T)$.
	If $\Phi(S_0)=0$, then the convergence order depends on the function $\Phi.$ As an extreme case, if $\Phi\equiv0,$ it is evident that these expectations are   of the order smaller than $\mathcal{O}(T).$
	For the ATM (i.e., $S_0=K$) call and put options with    
	$\Phi(x)=(x-K)_+$ and $\Phi(x)=(K-x)_+,$ we have
	$$\mathbb{E}^\mathbb{Q}\left[\Phi\left(S_0+S_0\sigma_A(T)Z\right)\right]=\frac{S_0\sigma_A(T)}{\sqrt{2\pi}}\sqrt{T}$$
	and
	$$\mathbb{E}^\mathbb{Q}\left[\Phi\left(S_0+S_0\sigma_E(T)Z\right)\right]=\frac{S_0\sigma_E(T)}{\sqrt{2\pi}}\sqrt{T}\,,$$
	which means that these are   of the order greater than $\mathcal{O}(T^\gamma).$
 
In \cite{PirjolDan2016SMAO}, the authors derive an approximation formula for $P_A(T)$ of order $e^{-\frac{\mathcal{I}}{T}+o(\frac{1}{T})}$ using large deviations theory, where $\mathcal{I}$ denotes the rate function.
Their result provides a more detailed characterization, as the approximation error of order $O(T)$ is much larger than the exponential term $e^{-O(\frac{1}{T})}$.
Consequently, their short-maturity implied volatility formula differs from ours and offers a finer asymptotic expansion.
However, in the ATM case where $\sigma(t,x) \equiv \sigma(x)$, their short-maturity implied volatility coincides with ours:
$$
\sigma_A(T) = \sqrt{\frac{1}{T^3} \int_0^T \sigma^2(t,S_0)(T-t)^2\,dt} = \frac{\sigma(S_0)}{\sqrt{3}}.$$
For further details, see \cite[Proposition 18 (iii)]{PirjolDan2016SMAO}.

Even when $1\le \gamma <2\,,$ we can estimate specific type of option prices by using a similar method presented above.
\begin{example}\label{ex: gamma price lipschitz}
 	Given any $K,\delta>0$ and $1\le \gamma <2\,,$ define the payoff function $\Phi$ by
	\begin{align}
	\Phi(x)=(x-K)^{\gamma}\mathbbm{1}_{\{K\le x<K+\delta\}}+\delta^{\gamma}\mathbbm{1}_{\{K+\delta\le x\}}\,.
	\end{align}
	If $S_0=K,$ then we obtain the   asymptotic equation
	\begin{equation}
	P_A(T)=\frac{1}{2}(S_0\sigma_A(T))^{\gamma}\,M(\gamma)T^{\frac{\gamma}{2}}+\mathcal{O}(T)\,,
	\end{equation}
	where $M(\gamma):=\mathbb{E}^\mathbb{Q}[\vert Z\vert^{\gamma}]$ with a standard normal variable $Z\,.$ If we replace $\sigma_A(T)$ with $\sigma_E(T)\,,$ we obtain the asymptotic result for the European option price $P_E(T)\,.$
\end{example}

\section{Preliminary estimates for delta values}
\label{sec:Short maturity asymptotic for a sensitivity}

In this section, we derive preliminary estimates of the short-maturity deltas for Asian and European options. The Asian and European delta values are defined respectively by
$$ \Delta_A(T):=\frac{\partial}{\partial S_0}P_A(T)\,,\quad \Delta_E(T):=\frac{\partial}{\partial S_0}P_E(T)\,. $$ 
We introduce six processes that are used to approximate the stock price process $(S_t)_{t\ge0}$. 
Define a process $(X_t)_{t\ge0}$ as a solution to 
\begin{equation}\label{eq:X_t}
dX_t=\sigma(t,X_t)X_t\,dW_t\,,\quad X_0=S_0>0\,
\end{equation}
and its first variation process $(Y_t)_{t\ge0}$  as 
\begin{equation}\label{eq:Y_t}
dY_t=\nu(t,X_t)Y_t\,dW_t\,,\quad Y_0=1\,.
\end{equation}
We also define  two 
geometric Gaussian processes $(\tilde{X}_t)_{t\ge0}$ and $(\tilde{Y}_t)_{t\ge0}$ as solutions to
\begin{equation}\label{eq:tilde x,y}
d\tilde{X}_t=\sigma(t,S_0)\tilde{X}_t\,dW_t\,,\quad \tilde{X}_0=S_0\,, \qquad
d\tilde{Y}_t=\nu(t,S_0)\tilde{Y}_t\,dW_t\,, \quad \tilde{Y}_0=1\,,
\end{equation}
and two 
Gaussian processes $(\hat{X}_t)_{t\ge0}$ and $(\hat{Y}_t)_{t\ge0}$ as solutions to
\begin{equation}\label{eq:hat x,y}
d\hat{X}_t=\sigma(t,S_0)S_0\,dW_t\,,\quad \hat{X}_0=S_0\,, \qquad d\hat{Y}_t=\nu(t,S_0)\,dW_t\,, \quad \hat{Y}_0=1\,.
\end{equation}
It is evident that the processes $X,Y,\tilde{X}$ and $\tilde{Y}$ are positive.

\subsection{Asian option deltas }\label{subsec:asian delta comp}

By Malliavin calculus, the Asian delta value can be represented by the weighted sum of the payoffs. 

The representation is given in the following proposition, whose proof can be found in \cite{benhamou2000application}.
We refer to \cite{NualartDavid1995TMca,nualart2009malliavin} for the definitions and basic properties of the Malliavin derivative and the Skorokhod integral.

\begin{proposition}\label{prop:asian delta mal} 
Let Assumption \ref{classical assumption} hold and  let $X$ and $Y$ be solutions to \eqref{eq:X_t} and \eqref{eq:Y_t}, respectively. 
  Then,
	\begin{align}
		\frac{\partial}{\partial S_0}\mathbb{E}^\mathbb{Q}\left[\Phi\left(\frac{1}{T}\int_0^T X_t\,dt\right)\right]
		&=\mathbb{E}^\mathbb{Q}\left[\Phi\left(\frac{1}{T}\int_0^TX_t\,dt\right)\delta\left(\frac{2{Y_{\cdot}}^2}{\sigma(\cdot,X_{\cdot})X_{\cdot}\int_0^T Y_t\,dt}\right)\right] \\
		&=\mathbb{E}^\mathbb{Q}\left[\Phi\left(\frac{1}{T}\int_0^T X_t\,dt\right)\delta\left(\frac{2{Y_{\cdot}}^2}{\sigma(\cdot,X_{\cdot})X_{\cdot}}\right)\frac{1}{\int_0^T Y_t\,dt}\right] \\
		&\quad-\mathbb{E}^\mathbb{Q}\left[\Phi\left(\frac{1}{T}\int_0^T X_t\,dt\right)\int_0^T \frac{2{Y_s}^2}{\sigma(s,X_s)X_s}D_{s}\left(\frac{1}{\int_0^T Y_t\,dt}\right)\,ds\right]\,,
	\end{align}
	where $\delta(\cdot)$ is the Skorokhod integral in $[0,T]$ and $D_s(\cdot)$ is the Malliavin derivative.
\end{proposition}

We define a process $(u_s)_{0\le s\le T}$ and a random variable $F$ by
\begin{equation}
\label{eqn:uF}  u_{s}:=\frac{2Y_s^2}{\sigma(s,X_s)X_s}\,,\; F:=\frac{1}{\int_0^T Y_t\,dt}\,.
\end{equation}
Similarly, we introduce 
a process  $(\hat{u}_s)_{0\le s\le T}$ 
and a random variable 
 $\hat{F},$ defined as
\begin{equation}
\label{eqn:hat_uF}  \hat{u}_s:=\frac{2{\hat{Y}_s}^2}{\sigma(s,\hat{X}_s)\hat{X}_s}\mathbbm{1}_{\{\hat{X}_s\ge\frac{S_0}{2}\}} \,, \; \hat{F}:=\frac{1}{\int_0^T \hat{Y}_t\,dt}\mathbbm{1}_{\{\frac{1}{T}\int_0^T \hat{Y}_t\,dt\ge\frac{1}{2}\}}\,, 
\end{equation}
where $\mathbbm{1}_A$ denotes the indicator function of the set $A$.  
We estimate the two expectations
\begin{equation}\label{eq:asian delta mal cal 3}
\mathbb{E}^\mathbb{Q}\left[\Phi\left(\frac{1}{T}\int_0^T \hat{X}_t\,dt\right)\delta(\hat{u})\hat{F}\right]\;\textnormal{ and }\;
\mathbb{E}^\mathbb{Q}\left[\Phi\left(\frac{1}{T}\int_0^T \hat{X}_t\,dt\right)\int_0^T \hat{u}_{s}(D_{s}^{*}\hat{F})\,ds\right]
\end{equation}
where
$$D_s^{*}\hat{F}:=D_s\left(\frac{1}{\int_0^T\hat{Y}_t\,dt}\right)\mathbbm{1}_{\{\frac{1}{T}\int_0^T\hat{Y_t}\,dt\ge \frac{1}{2}\}}\,.$$

If $\gamma<\frac{1}{2},$ the estimates in the following proposition are meaningless; however, for $\frac{1}{2}<\gamma\le 1,$ they provide us with the short-maturity estimate with the convergence rate $\gamma-\frac{1}{2}.$

\begin{proposition}\label{prop:Asian delta app 3}
	Under Assumptions \ref{classical assumption}
	and \ref{Holder payoff assumption}, as $T\rightarrow{0}$, we have
	\begin{equation}\label{eq:Asian delta app 3.1}
	\begin{aligned}
	&\mathbb{E}^\mathbb{Q}\left[\Phi\left(\frac{1}{T}\int_0^T \hat{X}_t\,dt\right)\delta(\hat{u})\hat{F}\right] \\ &=\mathbb{E}^\mathbb{Q}\left[\frac{\Phi(S_0+S_0\sigma_A(T)\sqrt{T}Z)}{S_0\sigma_A(T)\sqrt{T}}Z\right]
	-2\frac{\Phi(S_0)}{S_0}\frac{1}{T^2}\int_0^T\frac{\nu(s,S_0)}{\sigma(s,S_0)}(T-s)\,ds +\mathcal{O}(T^{\gamma-\frac{1}{2}})\,
	\end{aligned}
	\end{equation}
	and
	\begin{equation}\label{eq:Asian delta app 3.2}
	\mathbb{E}^\mathbb{Q}\left[\Phi\left(\frac{1}{T}\int_0^T \hat{X}_t\,dt\right)\int_0^T \hat{u}_{s}(D_{s}^{*}\hat{F})\,ds\right]
	=-2\frac{\Phi(S_0)}{S_0}\frac{1}{T^2}\int_0^T\frac{\nu(s,S_0)}{\sigma(s,S_0)}(T-s)\,ds +\mathcal{O}(T^{\gamma-\frac{1}{2}})\,,
	\end{equation}
	where $Z$ is a standard normal random variable.
\end{proposition}

\begin{proof}
We first prove \eqref{eq:Asian delta app 3.1}.   Define
	\begin{equation}
	\Delta_A^{*}(T):=\mathbb{E}^\mathbb{Q}\left[\Phi\left(\frac{1}{T}\int_0^T \hat{X}_t\,dt\right)\frac{1}{T}\delta(\hat{u})\right]\,,
	\end{equation}
then we have
	\begin{equation}
	\Delta_A^{*}(T)=\mathbb{E}^\mathbb{Q}\left[\frac{\Phi(S_0+S_0\sigma_A(T)\sqrt{T}Z)}{S_0\sigma_A(T)\sqrt{T}}Z\right]+\mathcal{O}(T^{\gamma-\frac{1}{2}})\,.
	\end{equation} 
	This can be proven as follows.
	From $\mathbb{E}^\mathbb{Q}\left[\delta(\hat{u})\right]=\mathbb{E}^\mathbb{Q}[\delta (\frac{2}{\sigma(\cdot, S_0)S_0})]=0 $, Lemma \ref{lem:hat u app} and Theorem 3.4.3 of \cite{zhang2017backward}, we obtain
	\begin{align}
	&\left\vert \Delta_A^{*}(T)-\mathbb{E}^\mathbb{Q}\left[\Phi\left(\frac{1}{T} \int_0^T\hat{X}_t\,dt\right)\frac{1}{T}\delta\left(\frac{2}{\sigma(\cdot,S_0)S_0}\right)\right] -\mathbb{E}^\mathbb{Q}\left[\Phi(S_0)\frac{1}{T}\delta(g)\right]\right\vert \\
	&\le \mathbb{E}^\mathbb{Q}\left[\left\vert \Phi\left(\frac{1}{T}\int_0^T\hat{X}_t\,dt\right)-\Phi(S_0)\right\vert\frac{1}{T}\vert\delta(g)\vert\right]\le C_1T^{\gamma-\frac{1}{2}}
	\end{align}
for some constant $C_1>0$
	where $g_s:=\hat{u}_s-\frac{2}{\sigma(s,S_0)S_0}\,.$ 
	From the stochastic Fubini theorem,
   	\begin{align}
	&\quad\mathbb{E}^\mathbb{Q}\left[\Phi\left(\frac{1}{T}\int_0^T\hat{X}_t\,dt\right)\frac{1}{T}\delta\left(\frac{2}{\sigma(\cdot,S_0)S_0}\right)\right] \\
	&=\mathbb{E}^\mathbb{Q}\left[\Phi\left(S_0+\frac{S_0}{T}\int_0^T\sigma(s,S_0)(T-s)\,dW_s\right)\frac{1}{T}\int_0^T\frac{2}{\sigma(s,S_0)S_0}\,dW_s\right]\\
    &=\mathbb{E}^\mathbb{Q}\left[\Phi\left(S_0 + \frac{S_0}{T}U\right)\frac{1}{T}V\right]\,
	\end{align}
 where 
\begin{align}
   U:=\int_0^T\sigma(s,S_0)(T-s)\,dW_s\,,\;V:=\int_0^T\frac{2}{\sigma(s,S_0)S_0}\,dW_s\,.
\end{align}  
Two random variables  $U$ and $
V - \frac{T^2/S_0}{\int_0^T\sigma^2(s,S_0)(T-s)^2\,ds}U$ are independent since 
they are jointly normal and \begin{align}
&\quad\text{Cov}\left(U, V - \frac{T^2/S_0}{\int_0^T\sigma^2(s,S_0)(T-s)^2,ds}U \right) \\
&= \mathbb{E}[U\,V] - \frac{T^2/S_0}{\int_0^T\sigma^2(s,S_0)(T-s)^2\,ds} \mathbb{E}[U^2] \\
&= \int_0^T \sigma(s,S_0)(T-s) \frac{2}{\sigma(s,S_0)S_0} \,ds - \frac{T^2/S_0}{\int_0^T\sigma^2(s,S_0)(T-s)^2\,ds} \int_0^T \sigma^2(s,S_0)(T-s)^2 \,ds \\
&= \int_0^T \frac{2(T-s)}{S_0} \,ds - \frac{T^2}{S_0} = 0\,.
\end{align}
 Therefore,  the desired result follows from 
\begin{align}
&\quad\mathbb{E}^\mathbb{Q}\left[\Phi\left(S_0 + \frac{S_0}{T}U\right)\frac{1}{T}V\right] \\
&= \mathbb{E}^\mathbb{Q}\left[\Phi\left(S_0 + \frac{S_0}{T}U\right)\frac{1}{T}\left(\frac{T^2/S_0}{\int_0^T \sigma^2(s,S_0)(T - s)^2 ds}\right)U\right] \\
&= \mathbb{E}^\mathbb{Q}\left[\Phi\left(S_0 + S_0\sigma_A(T)\sqrt{T}Z\right)\frac{Z}{S_0\sigma_A(T)\sqrt{T}}\right]\,.
\end{align}

	We now show that 
	\begin{equation}\label{eqn:d}
	\Delta_A^{*}(T)-\mathbb{E}^\mathbb{Q}\left[\Phi\left(\frac{1}{T}\int_0^T \hat{X}_t\,dt\right)\delta(\hat{u})\hat{F}\right] 
	=\Phi(S_0)\mathbb{E}^\mathbb{Q}\left[\delta(\hat{u})\hat{F}\left(\frac{1}{T}\int_0^T\hat{Y}_t\,dt-1\right)\right]+\mathcal{O}(T^{\gamma-\frac{1}{2}})\,.
	\end{equation}
Observe that
	\begin{align}
	\Delta_A^{*}(T)&=\mathbb{E}^\mathbb{Q}\left[\Phi\left(\frac{1}{T}\int_0^T \hat{X}_t\,dt\right)\frac{1}{T}\delta(\hat{u})\mathbbm{1}_{\{\frac{1}{T}\int_0^T\hat{Y}_t\,dt\ge\frac{1}{2}\}}\right]+\mathcal{O}(T^{\gamma-\frac{1}{2}})\,.
	\end{align}
    This follows by assuming $\Phi(0)=0$ (otherwise replacing $\Phi$ with $\Phi(\cdot)-\Phi(0)$), and applying Lemma \ref{lem:hat u app}, \eqref{eq:hatY vanish}, together with the estimate $N(-\frac{\sqrt{3}}{2\alpha\sqrt{T}})=o(T^q)$ for any $q>0$ as $T\rightarrow{0}$.  
	Since
	\begin{align}
&\quad	\Delta_A^{*}(T)-\mathbb{E}^\mathbb{Q}\left[\Phi\left(\frac{1}{T}\int_0^T \hat{X}_t\,dt\right)\delta(\hat{u})\hat{F}\right]\\
    &=\mathbb{E}^\mathbb{Q}\left[\Phi\left(\frac{1}{T}\int_0^T \hat{X}_t\,dt\right)\delta(\hat{u})\hat{F}\left(\frac{1}{T}\int_0^T\hat{Y}_t\,dt-1\right)\right]+\mathcal{O}(T^{\gamma-\frac{1}{2}})\,,
	\end{align}
	it suffices to show that
	\begin{align}
	&\quad\mathbb{E}^\mathbb{Q}\left[\Phi\left(\frac{1}{T}\int_0^T \hat{X}_t\,dt\right)\delta(\hat{u})\hat{F}\left(\frac{1}{T}\int_0^T\hat{Y}_t\,dt-1\right)\right]\\
    &=\Phi(S_0)\mathbb{E}^\mathbb{Q}\left[\delta(\hat{u})\hat{F}\left(\frac{1}{T}\int_0^T\hat{Y}_t\,dt-1\right)\right]+\mathcal{O}(T^{\gamma-\frac{1}{2}})\,.
	\end{align}
From the inequality $T\hat{F}\le 2$, Lemma \ref{lem:hat u app} and Theorem 3.4.3 of \cite{zhang2017backward}, we have 
	\begin{align}
	\mathbb{E}^\mathbb{Q}\left[\delta(\hat{u})\hat{F}\left(\frac{1}{T}\int_0^T\hat{Y}_t\,dt-1\right)\right]
	=\mathbb{E}^\mathbb{Q}\left[\frac{1}{T}\delta(\hat{u})\left(\frac{1}{T}\int_0^T\hat{Y}_t\,dt-1\right)\right]+\mathcal{O}(T^{\gamma-\frac{1}{2}})\,.
	\end{align}
    Then 
	\begin{align}
	\mathbb{E}^\mathbb{Q}\left[\frac{1}{T}\delta(\hat{u})\left(\frac{1}{T}\int_0^T\hat{Y}_t\,dt-1\right)\right]
	=\frac{2}{S_0}\frac{1}{T^2}\int_0^T\frac{\nu(s,S_0)}{\sigma(s,S_0)}(T-s)\,ds+\mathcal{O}(T^{\gamma-\frac{1}{2}})\,,
	\end{align} 
which is obtained by Lemma \ref{lem:hat u app}, 
    	\begin{align}
	\mathbb{E}^\mathbb{Q}\left[\frac{1}{T}\delta(\hat{u})\left(\frac{1}{T}\int_0^T\hat{Y}_t\,dt-1\right)\right]
	=\mathbb{E}^\mathbb{Q}\left[\frac{1}{T}\delta\left(\frac{2}{\sigma(\cdot,S_0)S_0}\right)\left(\frac{1}{T}\int_0^T\hat{Y}_t\,dt-1\right)\right]+\mathcal{O}(T^{\gamma-\frac{1}{2}})\,.
	\end{align}
	and	\begin{align}
	&\mathbb{E}^\mathbb{Q}\left[\frac{1}{T}\delta\left(\frac{2}{\sigma(\cdot,S_0)S_0}\right)\left(\frac{1}{T}\int_0^T\hat{Y}_t\,dt-1\right)\right]\\
	&=\mathbb{E}^\mathbb{Q}\left[\left(\frac{1}{T}\int_0^T\frac{2}{\sigma(s,S_0)S_0}\,dW_s\right)\left(\frac{1}{T}\int_0^T\nu(s,S_0)(T-s)\,dW_s\right)\right]\\
	&=\frac{2}{S_0}\frac{1}{T^2}\int_0^T\frac{\nu(s,S_0)}{\sigma(s,S_0)}(T-s)\,ds\,.
	\end{align}
	This completes the proof of \eqref{eq:Asian delta app 3.1}.

 We now prove \eqref{eq:Asian delta app 3.2}. We divide the proof into several steps.
	From Lemma \ref{lem:Ds bound} and the definition of $\hat{u}_s$, we have
	\begin{align}
	&\left\vert\mathbb{E}^\mathbb{Q}\left[\Phi\left(\frac{1}{T}\int_0^T \hat{X}_t\,dt\right)\int_0^T \hat{u}_{s}(D_{s}^{*}\hat{F})\,ds\right]-\Phi(S_0)\mathbb{E}^\mathbb{Q}\left[\int_0^T \hat{u}_{s}(D_{s}^{*}\hat{F})\,ds\right]\right\vert \\
	&\le \beta\left(\frac{1}{T}\int_0^T\mathbb{E}^\mathbb{Q}[\vert\hat{X}_t-S_0\vert^2]\,dt\right)^{\frac{1}{2}}\left(\frac{1}{T}\int_0^T\mathbb{E}^\mathbb{Q}[\hat{u}_s^2( TD_s^{*}\hat{F})^2]\,ds\right)^{\frac{1}{2}}\le C_2T^{\gamma-\frac{1}{2}}
	\end{align}
for some  constant $C_2>0$.
	Then
	\begin{align}
	\mathbb{E}^\mathbb{Q}\left[\Phi\left(\frac{1}{T}\int_0^T \hat{X}_t\,dt\right)\int_0^T \hat{u}_{s}(D_{s}^{*}\hat{F})\,ds\right]
	=\Phi(S_0)\mathbb{E}^\mathbb{Q}\left[\int_0^T \hat{u}_{s}(D_{s}^{*}\hat{F})\,ds\right]+\mathcal{O}(T^{\gamma-\frac{1}{2}})\,.
	\end{align}		
	Similarly,
	\begin{align}
	&\left\vert\mathbb{E}^\mathbb{Q}\left[\int_0^T \hat{u}_{s}(D_{s}^{*}\hat{F})\,ds\right]
	-\mathbb{E}^\mathbb{Q}\left[\int_0^T \hat{u}_{s}(D_{s}^{*}\hat{F})\,ds\left(\frac{1}{T}\int_0^T\hat{Y}_t\,dt\right)^2\right]\right\vert \\
	&\le\left(\frac{1}{T}\int_0^T\mathbb{E}^\mathbb{Q}[\vert\hat{u}_s\vert^2\vert TD_s^{*}\hat{F}\vert^2]\,ds\right)^{\frac{1}{2}}\left(\mathbb{E}^\mathbb{Q}\left[\left(\left(\frac{1}{T}\int_0^T\hat{Y}_t\,dt\right)^2-1\right)^2\right]\right)^{\frac{1}{2}}\le C_3T^{\gamma-\frac{1}{2}}
	\end{align}
for some constant $C_3>0$.		
	Thus,
	\begin{align}
	\mathbb{E}^\mathbb{Q}\left[\int_0^T \hat{u}_{s}(D_{s}^{*}\hat{F})\,ds\right]
	=\mathbb{E}^\mathbb{Q}\left[\int_0^T \hat{u}_{s}(D_{s}^{*}\hat{F})\,ds\left(\frac{1}{T}\int_0^T\hat{Y}_t\,dt\right)^2\right]+\mathcal{O}(T^{\gamma-\frac{1}{2}})\,.
	\end{align}
	From \eqref{p-th moment of g} and Lemma \ref{lem:Ds bound}, we have
	\begin{align}
	&\left\vert\mathbb{E}^\mathbb{Q}\left[\int_0^T \hat{u}_{s}(D_{s}^{*}\hat{F})\,ds\left(\frac{1}{T}\int_0^T\hat{Y}_t\,dt\right)^2\right]
	-\mathbb{E}^\mathbb{Q}\left[\int_0^T\frac{2}{\sigma(s,S_0)S_0}(D_{s}^{*}\hat{F})\,ds\left(\frac{1}{T}\int_0^T\hat{Y}_t\,dt\right)^2\right]\right\vert \\
	&\le \left(\frac{1}{T}\int_0^T\mathbb{E}^\mathbb{Q}[g_s^2(TD_s^{*}\hat{F})^2]\,ds\right)^{\frac{1}{2}}\left(\frac{1}{T}\int_0^T\mathbb{E}^\mathbb{Q}[\hat{Y}_t^4]\,dt\right)^{\frac{1}{2}}\le C_4T^{\gamma-\frac{1}{2}}
	\end{align}
for some  constant $C_4>0$	where $g_s:=\hat{u}_s-\frac{2}{\sigma(s,S_0)S_0}.$
	Then 
	\begin{align}
	&\quad\mathbb{E}^\mathbb{Q}\left[\int_0^T \hat{u}_{s}(D_{s}^{*}\hat{F})\,ds\left(\frac{1}{T}\int_0^T\hat{Y}_t\,dt\right)^2\right] \\
    &=\mathbb{E}^\mathbb{Q}\left[\int_0^T\frac{2}{\sigma(s,S_0)S_0}(D_{s}^{*}\hat{F})\,ds\left(\frac{1}{T}\int_0^T\hat{Y}_t\,dt\right)^2\right]+\mathcal{O}(T^{\gamma-\frac{1}{2}})\,.
	\end{align}
	From the definition of $D_s^{*}\hat{F}$ and the estimation of $D_s\hat{Y}_t$ in \eqref{proof:DtildeY, DhatY}, it follows that
	\begin{equation}
	D_s^{*}\hat{F}=-\frac{\int_0^TD_s\hat{Y}_t\,dt}{\left(\int_0^T\hat{Y}_t\,dt\right)^2}\mathbbm{1}_{\{\frac{1}{T}\int_0^T\hat{Y}_t\,dt\ge\frac{1}{2}\}}
	=-\frac{1}{T^2}\frac{\nu(s,S_0)(T-s)}{\left(\frac{1}{T}\int_0^T\hat{Y}_t\,dt\right)^2}\mathbbm{1}_{\{\frac{1}{T}\int_0^T\hat{Y}_t\,dt\ge\frac{1}{2}\}}\,.
	\end{equation}
	Using this identity, and $\mathbb{Q}\{\frac{1}{T}\int_0^T\hat{Y}_t\,dt<\frac{1}{2}\}=o(T^q)$ for any $q>0\,$ as $T\rightarrow{0}\,,$ 
	we have 
	\begin{align}
	&\quad\mathbb{E}^\mathbb{Q}\left[\int_0^T\frac{2}{\sigma(s,S_0)S_0}(D_{s}^{*}\hat{F})\,ds\left(\frac{1}{T}\int_0^T\hat{Y}_t\,dt\right)^2\right]\\
	&=-\frac{2}{S_0}\frac{1}{T^2}\int_0^T\frac{\nu(s,S_0)}{\sigma(s,S_0)}(T-s)\,ds+\mathcal{O}(T^{\gamma-\frac{1}{2}})\,.
	\end{align}
	Combining this with \eqref{eqn:d}, we finally obtain \eqref{eq:Asian delta app 3.2}. This completes the proofs.
\end{proof}

Similarly, we estimate two expectations in the following proposition.
 
\begin{proposition}\label{prop:Asian delta app 1}
	Under Assumptions \ref{classical assumption}
	and \ref{Holder payoff assumption}, as $T\rightarrow{0}$, we have
	\begin{equation}\label{eq:asian delta mal cal 1}
		\mathbb{E}^\mathbb{Q}\left[\Phi\left(\frac{1}{T}\int_0^T X_t\,dt\right)\delta(u)F\right]=\mathbb{E}^\mathbb{Q}\left[\Phi\left(\frac{1}{T}\int_0^T \hat{X}_t\,dt\right)\delta(\hat{u})\hat{F}\right]+\mathcal{O}(T^{\gamma-\frac{1}{2}})
	\end{equation}
	and
		\begin{equation}\label{eq:Asian delta app 2.1}
	\begin{aligned}
	\mathbb{E}^\mathbb{Q}\left[\Phi\left(\frac{1}{T}\int_0^T X_t\,dt\right)\int_0^T u_{s}(D_{s}F)\,ds\right] =\mathbb{E}^\mathbb{Q}\left[\Phi\left(\frac{1}{T}\int_0^T \hat{X}_t\,dt\right)\int_0^T \hat{u}_{s}(D_{s}^{*}\hat{F})\,ds\right]
	+\mathcal{O}(T^{\gamma-\frac{1}{2}}).
	\end{aligned}
	\end{equation}
\end{proposition}

\begin{proof} Since the proof is similar to that of Proposition \ref{prop:Asian delta app 3}, we only provide a sketch of the main arguments.
	We may assume that $\Phi(0)=0$ by considering the translation $\Phi(\cdot)-\Phi(0)$  otherwise. Observe that
	\begin{align}
		&\quad\Phi\left(\frac{1}{T}\int_0^T X_t\,dt\right)\delta(u)F-\Phi\left(\frac{1}{T}\int_0^T \hat{X}_t\,dt\right)\delta(\hat{u})\hat{F} \\
		&= \frac{1}{T}\left(\Phi\left(\frac{1}{T}\int_0^T X_t\,dt\right)-\Phi\left(\frac{1}{T}\int_0^T \hat{X}_t\,dt\right)\right)\delta(u)TF \\
		&\quad+\frac{1}{T}\Phi\left(\frac{1}{T}\int_0^T \hat{X}_t\,dt\right)  \delta(u-\hat{u}) TF+\frac{1}{T}\Phi\left(\frac{1}{T}\int_0^T \hat{X}_t\,dt\right)\delta(\hat{u}) \left(TF-T\hat{F}\right) \,.
	\end{align}
	From Lemmas \ref{lem:x,y close}, \ref{lem:dummy}, \ref{lem:u,TF app}, and the fact that the Skorokhod integral of $u$ coincides with  the It\^o integral
	whenever $(u_s)_{0\le s\le T}$ is adapted to the Brownian filtration $(\mathcal{F}_s^W)_{0\le s\le T}$, we obtain \eqref{eq:asian delta mal cal 1}.

	Similarly, 	we have
	\begin{align}
	&\quad\Phi\left(\frac{1}{T}\int_0^T X_t\,dt\right)\int_0^T u_{s}(D_{s}F)\,ds -\Phi\left(\frac{1}{T}\int_0^T \hat{X}_t\,dt\right)\int_0^T \hat{u}_{s}(D_{s}^{*}\hat{F})\,ds \\
	&= \left(\Phi\left(\frac{1}{T}\int_0^T X_t\,dt\right)-\Phi\left(\frac{1}{T}\int_0^T \hat{X}_t\,dt\right)\right)\frac{1}{T}\int_0^T u_s \left(TD_sF\right) \,ds \\
	&\quad+\Phi\left(\frac{1}{T}\int_0^T \hat{X}_t\,dt\right)\frac{1}{T}\int_0^T \left(u_s-\hat{u}_s\right) TD_sF\, ds\\
    &\quad+\Phi\left(\frac{1}{T}\int_0^T\hat{X}_t\,dt\right)\frac{1}{T}\int_0^T \hat{u}_s\left( TD_sF-TD_{s}^{*}\hat{F}\right)\,ds.
	\end{align}
	Applying  Lemmas \ref{lem:x,y close}, \ref{lem:dummy}, \ref{lem:u,TF app}, \ref{lem:Ds bound}, \ref{lem:Ds close}, we obtain \eqref{eq:Asian delta app 2.1}. 
\end{proof}

\subsection{European option deltas } \label{subsec:european delta comp}

Using Malliavin calculus, we represent the European delta as a weighted average of the payoff function.
Refer to \cite{NualartDavid1995TMca,benhamou2000application} for the following proposition.

\begin{proposition} \label{prop:euro}
	Let Assumption \ref{classical assumption} hold. For the processes $X$ and $Y$
	defined in \eqref{eq:X_t} and \eqref{eq:Y_t} respectively, we have
	\begin{align}
		\frac{\partial}{\partial S_0}\mathbb{E}^\mathbb{Q}[\Phi(X_T)]&=\frac{1}{S_0T}\,\mathbb{E}^\mathbb{Q}\left[\Phi\left(X_T\right)\delta\left(\frac{Y_{\cdot}}{\sigma(\cdot,X_{\cdot})X_{\cdot}}\right)\frac{X_T}{Y_T}\right] -\frac{1}{S_0}\,\mathbb{E}^\mathbb{Q}[\Phi(X_T)] \\
		&+\frac{1}{S_0T}\,\mathbb{E}^\mathbb{Q}\left[\Phi(X_T)\int_0^T \frac{Y_s}{\sigma(s,X_s)X_s}\frac{X_T(D_{s}Y_T)}{{Y_T}^2}ds\right],
	\end{align}
	where $\delta(\cdot)$ denotes the Skorokhod integral in $[0,T]$ and $D_s(\cdot)$ is the Malliavin derivative.
\end{proposition}

Define stochastic processes $(h_s)_{0\le s\le T}$ and $(H_s)_{0\le s\le T}$ together with a random variable $G$ by
\begin{align}\label{eqn:orig}
	h_s:=\frac{Y_s}{\sigma(s,X_s)X_s}\,, \;
	H_s:=\frac{X_T(D_{s}Y_T)}{Y_T^2}\,, \;
	G:=\frac{X_T}{Y_T}\,.
\end{align}
In addition, we introduce auxiliary processes $(\hat{h}_s)_{0\le s\le T}$, $(\hat{H}_s)_{0\le s\le T}$ and an auxiliary random variable $\hat{G}$ as
\begin{align}\label{eqn:aux}
	\hat{h}_s:=\frac{\hat{Y}_s}{\sigma(s,\hat{X}_s)\hat{X}_s}\mathbbm{1}_{\{\hat{X}_s\ge\frac{S_0}{2}\}}\,, \;
	\hat{H}_s:=\frac{\hat{X}_T(D_{s}\hat{Y}_T)}{{\hat{Y}_T}^2}\mathbbm{1}_{\{\hat{Y}_T\ge\frac{1}{2}\}}\,,\;\hat{G}:=\frac{\hat{X}_T}{\hat{Y}_T}\mathbbm{1}_{\{\hat{Y}_T\ge\frac{1}{2}\}}\,.
\end{align}
Since the proof of the following proposition is similar to that of Proposition \ref{prop:Asian delta app 3}, we omit it.

\begin{proposition}\label{prop:european delta app 1}
	Under Assumptions \ref{classical assumption}
	and \ref{Holder payoff assumption}, as $T\rightarrow{0}$, we have 
	\begin{align}
	\frac{1}{S_0T}\,\mathbb{E}^\mathbb{Q}\Big[\Phi(X_T)\delta(h)G\Big]=\frac{1}{S_0T}\,\mathbb{E}^\mathbb{Q}\left[\Phi(\hat{X}_T)\delta(\hat{h})\hat{G}\right]+\mathcal{O}(T^{\gamma-\frac{1}{2}})
	\end{align}
	and 
	\begin{align}
		\frac{1}{S_0T}\,\mathbb{E}^\mathbb{Q}\left[\Phi(X_T)\int_0^T h_s H_s\,ds\right]=\frac{1}{S_0T}\,\mathbb{E}^\mathbb{Q}\left[\Phi(\hat{X}_T)\int_0^T \hat{h}_s \hat{H}_s\,ds\right]+\mathcal{O}(T^{\gamma-\frac{1}{2}}).
	\end{align}
	In particular, \begin{equation}\label{eq:european delta mal cal}
	\Delta_E(T)=\frac{1}{S_0T}\,\mathbb{E}^\mathbb{Q}\left[\Phi(\hat{X}_T)\delta(\hat{h})\hat{G}\right]-\frac{\Phi(S_0)}{S_0}+\frac{1}{S_0T}\,\mathbb{E}^\mathbb{Q}\left[\Phi(\hat{X}_T)\int_0^T \hat{h}_s \hat{H}_s\,ds\right]+\mathcal{O}(T^{\gamma-\frac{1}{2}})\,.
	\end{equation}
\end{proposition}

The following proposition provides estimates for the two expectations in \eqref{eq:european delta mal cal}.
\begin{proposition}\label{prop:european delta app 2}
	Under Assumptions \ref{classical assumption}
	and \ref{Holder payoff assumption}, as $T\rightarrow{0}$, we have
	\begin{equation}\label{eq:european delta app 2.1}
		\begin{aligned}
			&\frac{1}{S_0T}\mathbb{E}^\mathbb{Q}\left[\Phi(\hat{X}_T)\delta(\hat{h})\hat{G}\right] \\
			&=\mathbb{E}^\mathbb{Q}\left[\frac{\Phi(S_0+S_0\sigma_E(T)\sqrt{T}Z)}{S_0\sigma_E(T)\sqrt{T}}Z\right]
			+\frac{\Phi(S_0)}{S_0T}\int_0^T\frac{\sigma(s,S_0)-\nu(s,S_0)}{\sigma(s,S_0)}\,ds+\mathcal{O}(T^{\gamma-\frac{1}{2}})\,
		\end{aligned}
	\end{equation}
	and
	\begin{equation}\label{eq:european delta app 2.2}
		\frac{1}{S_0T}\,\mathbb{E}^\mathbb{Q}\left[\Phi(\hat{X}_T)\int_0^T \hat{h}_s \hat{H}_s\,ds\right]
		=\frac{\Phi(S_0)}{S_0T}\int_0^T\frac{\nu(s,S_0)}{\sigma(s,S_0)}\,ds+\mathcal{O}(T^{\gamma-\frac{1}{2}})\,,
	\end{equation}
	where $Z$ is a standard normal random variable.
\end{proposition}

\begin{proof}
	We first prove \eqref{eq:european delta app 2.1}. From  Lemma \ref{lem:hat h app} and $\mathbb{E}^\mathbb{Q}[\vert\hat{X}_T-S_0\vert^p]=\mathcal{O}(T^{\frac{p}{2}})$, $\mathbb{E}^\mathbb{Q}[\hat{G}^p]=\mathcal{O}(1)$,  $\mathbb{E}^\mathbb{Q}[\vert\hat{G}-S_0\vert^p]=\mathcal{O}(T^{\frac{p}{2}})$, $\mathbb{E}^\mathbb{Q}[\vert\delta(\frac{2}{\sigma(\cdot,S_0)S_0})\vert^p]=\mathcal{O}(T^{\frac{p}{2}})$ for $p>0$, we have 
	\begin{align}
		&\frac{1}{S_0T}\mathbb{E}^\mathbb{Q}\left[\Phi(\hat{X}_T)\delta(\hat{h})\hat{G}\right]
		-\frac{1}{S_0T}\mathbb{E}^\mathbb{Q}\left[\Phi(S_0)\delta(\hat{h})\hat{G}\right] \\
		&=\frac{1}{S_0T}\mathbb{E}^\mathbb{Q}\left[\left(\Phi(\hat{X}_T)-\Phi(S_0)\right)\delta\left(\frac{1}{\sigma(\cdot,S_0)S_0}\right)\hat{G}\right]+\mathcal{O}(T^{\gamma-\frac{1}{2}}) \label{proof:european delta 1} \\ 
		&=\frac{1}{S_0T}\mathbb{E}^\mathbb{Q}\left[\left(\Phi(\hat{X}_T)-\Phi(S_0)\right)\delta\left(\frac{1}{\sigma(\cdot,S_0)S_0}\right)S_0\right]+\mathcal{O}(T^{\gamma-\frac{1}{2}}) \label{proof:european delta 2}\\ 
		&=\mathbb{E}^\mathbb{Q}\left[\frac{\Phi(S_0+S_0\sigma_E(T)\sqrt{T}Z)}{S_0\sigma_E(T)\sqrt{T}}Z\right]+\mathcal{O}(T^{\gamma-\frac{1}{2}})\,. \label{proof:european delta 3}
	\end{align}
 Similarly, from $\mathbb{E}^\mathbb{Q}[\delta(\hat{h})]=0$, $\mathbbm{E}^\mathbb{Q}[\vert \hat{G}-S_0\vert^p]=\mathcal{O}(T^{\frac{p}{2}})$, $\mathbb{E}^\mathbb{Q}[\vert\hat{Y}_T-1\vert^p]=\mathcal{O}(T^{\frac{p}{2}})$  $\mathbb{Q}\{\hat{Y}_T<\frac{1}{2}\}=o(T^p)$ for  $p>0$, the It\^o isometry and Lemma \ref{lem:hat h app}, we have
	\begin{align}
		\frac{\Phi(S_0)}{S_0T}\mathbb{E}^\mathbb{Q}\left[\delta(\hat{h})\hat{G}\right]
		&=\frac{\Phi(S_0)}{S_0T}\mathbb{E}^\mathbb{Q}\left[\delta(\hat{h})(\hat{G}-S_0)\right] \label{proof:european delta 4}\\
		&=\frac{\Phi(S_0)}{S_0T}\mathbb{E}^\mathbb{Q}\left[\delta\left(\frac{1}{\sigma(\cdot,S_0)S_0}\right)(\hat{G}-S_0)\right]+\mathcal{O}(T^{\gamma-\frac{1}{2}}) \label{proof:european delta 5}\\
		&=\frac{\Phi(S_0)}{S_0T}\mathbb{E}^\mathbb{Q}\left[\delta\left(\frac{1}{\sigma(\cdot,S_0)S_0}\right)(\hat{G}-S_0)\hat{Y}_T\right]+\mathcal{O}(T^{\gamma-\frac{1}{2}}) \label{proof:european delta 6}\\
		&=\frac{\Phi(S_0)}{S_0T}\mathbb{E}^\mathbb{Q}\left[\delta\left(\frac{1}{\sigma(\cdot,S_0)S_0}\right)(\hat{X}_T-S_0\hat{Y}_T)\right]+\mathcal{O}(T^{\gamma-\frac{1}{2}}) \label{proof:european delta 7}\\
		&=\frac{\Phi(S_0)}{S_0T}\int_0^T\frac{\sigma(s,S_0)-\nu(s,S_0)}{\sigma(s,S_0)}\,ds+\mathcal{O}(T^{\gamma-\frac{1}{2}})\,. \label{proof:european delta 8}
	\end{align}
By combining \eqref{proof:european delta 3} and \eqref{proof:european delta 8}, we obtain \eqref{eq:european delta app 2.1}.

We now verify \eqref{eq:european delta app 2.2}. From $\underset{0\le s\le T}{\sup}\,\mathbb{E}^\mathbb{Q}[\vert\hat{H}_s\vert^p]=\mathcal{O}(1)$,  $\mathbb{E}^\mathbb{Q}[\vert\hat{h}_s-\frac{1}{\sigma(s,S_0)S_0}\vert^p]=\mathcal{O}(s^{\frac{p}{2}})$ as $T\rightarrow{0}$, we have
	\begin{align}
		\frac{1}{T}\,\mathbb{E}^\mathbb{Q}\left[\Phi(\hat{X}_T)\int_0^T \hat{h}_s \hat{H}_s\,ds\right]
		&=\frac{1}{T}\,\mathbb{E}^\mathbb{Q}\left[\Phi(\hat{X}_T)\int_0^T \frac{1}{\sigma(s,S_0)S_0} \hat{H}_s\,ds\right]+\mathcal{O}(T^{\gamma-\frac{1}{2}}) \label{proof:european delta 9}\\
		&=\frac{1}{T}\,\mathbb{E}^\mathbb{Q}\left[\Phi(\hat{X}_T)\int_0^T \frac{1}{\sigma(s,S_0)S_0} \hat{H}_s\,ds\,\hat{Y}_T^2\right]+\mathcal{O}(T^{\gamma-\frac{1}{2}}) \label{proof:european delta 10}\\
		&=\frac{1}{T}\,\mathbb{E}^\mathbb{Q}\left[\Phi(\hat{X}_T)\int_0^T \frac{\nu(s,S_0)}{\sigma(s,S_0)S_0}\,ds\,\hat{X}_T\right]+\mathcal{O}(T^{\gamma-\frac{1}{2}}) \label{proof:european delta 11} \\
		&=\frac{\Phi(S_0)}{T}\int_0^T\frac{\nu(s,S_0)}{\sigma(s,S_0)}\,ds+\mathcal{O}(T^{\gamma-\frac{1}{2}})\label{proof:european delta 12}\,.
	\end{align}
	The proof for \eqref{eq:european delta app 2.2} is thus complete.
\end{proof}

\section{Short-maturity delta values}
\label{sec:delta_main}

We now present the short-maturity asymptotics for the deltas of Asian and European options. Recall that the Asian and European delta values are defined by
$$ \Delta_A(T):=\frac{\partial}{\partial S_0}P_A(T)\,,\quad \Delta_E(T):=\frac{\partial}{\partial S_0}P_E(T)\,. $$
Our main objective is to derive the short-maturity asymptotics for $\Delta_A(T)$ and $\Delta_E(T)$.
Refer to \cite{NualartDavid1995TMca,benhamou2000application} for the following proposition.

\begin{proposition}\label{prop:asian delta mal holder}
Let Assumption \ref{classical assumption} hold. 	For the process $S$ defined in \eqref{eq:S_t},  we have
	\begin{align}
	\Delta_A(T)
	&=e^{-rT}\mathbb{E}^\mathbb{Q}\left[\Phi\left(\frac{1}{T}\int_0^T S_t\,dt\right)\delta\left(\frac{2{Z_{\cdot}}^2}{\sigma(\cdot,S_{\cdot})S_{\cdot}}\right)\frac{1}{\int_0^T Z_t\,dt}\right] \\
	&\quad-e^{-rT}\mathbb{E}^\mathbb{Q}\left[\Phi\left(\frac{1}{T}\int_0^T S_t\,dt\right)\int_0^T \frac{2{Z_u}^2}{\sigma(u,S_u)S_u}D_{u}\left(\frac{1}{\int_0^T Z_t\,dt}\right)\,du\right]\,,
	\end{align}
	where $\delta(\cdot)$ is the Skorokhod integral in $[0,T]$, $D_s(\cdot)$ is the Malliavin derivative, and  $Z$ is the unique solution to the SDE 
	\begin{equation}\label{eq:process z}
	dZ_t=(r-q)Z_t\,dt+\nu(t,S_t)Z_t\,dW_t\,,\; Z_0=1.\,
	\end{equation}    
\end{proposition}

\begin{theorem}\label{thm:delta formula holder}
	Under Assumptions \ref{classical assumption} and \ref{Holder payoff assumption}, as $T\rightarrow{0}$, we have
	\begin{equation}\label{eqn:delta}
	\Delta_A(T)=\mathbb{E}^\mathbb{Q}\left[\frac{\Phi(S_0+S_0\sigma_A(T)\sqrt{T}Z)}{S_0\sigma_A(T)\sqrt{T}}Z\right]+\mathcal{O}(T^{\gamma-\frac{1}{2}})\,
	\end{equation}
	and
	\begin{equation}\label{eqn:delta_e}
	\Delta_E(T)=\mathbb{E}^\mathbb{Q}\left[\frac{\Phi(S_0+S_0\sigma_E(T)\sqrt{T}Z)}{S_0\sigma_E(T)\sqrt{T}}Z\right]+\mathcal{O}(T^{\gamma-\frac{1}{2}}) \,.
	\end{equation} 
\end{theorem}

\begin{proof} 
We prove the asymptotic behavior only for  the Asian delta $\Delta_A(T)$.
The result for the European delta $\Delta_E(T)$ follows by the same argument.    
	For notational simplicity, define the process $(v_t)_{0\le t\le T}$ and the random variable $F_Z$  as
	\begin{equation}
	v_t:=\frac{2Z^2_t}{\sigma(t,S_t)S_t}\,,\quad F_Z:=\frac{1}{\frac{1}{T}\int_0^T Z_t\,dt}\,.
	\end{equation}
	Recall from \eqref{eqn:uF} that $u_t=\frac{2Y^2_t}{\sigma(t,X_t)X_t}$ and $TF=\frac{1}{\frac{1}{T}\int_0^T Y_t\,dt}$.
	From Lemma \ref{lem:x,y close}, we have $\mathbb{E}^\mathbb{Q}[\vert S_t-X_t\vert^p]=\mathcal{O}(t^p)$ and $\mathbb{E}^\mathbb{Q}[\vert Z_t-Y_t\vert^p]=\mathcal{O}(t^p)$ for any $p>0$. 
	By the uniform Lipschitz continuity and boundedness conditions in Assumption \ref{classical assumption}, combined with H\"older's inequality, an argument analogous to the proof of Lemma \ref{lem:u,TF app} yields
	\begin{equation}
	\mathbb{E}^\mathbb{Q}\left[\left\vert v_t-u_t\right\vert^p\right]=\mathcal{O}(t^p)\,, \quad \mathbb{E}^\mathbb{Q}\left[\left\vert F_Z-TF\right\vert^p\right]=\mathcal{O}(T^p)
	\end{equation} 
	as $t, T\rightarrow{0}$ for any $p>0$. 
	
	Furthermore, by comparing the Malliavin derivatives
	\begin{equation}
	D_uS_l=\frac{Z_l}{Z_u}\sigma(u,S_u)S_u\mathbbm{1}_{\{u\le l\}}
	\end{equation}
	and  
	\begin{equation}
	D_uZ_t=Z_t\left[\nu(u,S_u)-\int_0^t\nu(l,S_l)\rho(l,S_l)D_uS_l\,dl+\int_0^t\rho(l,S_l)D_uS_l\,dW_l\right]\mathbbm{1}_{\{u\le t\}}
	\end{equation}
	with $D_uX_l$ and $D_uY_t$ (see \eqref{eq:DsXt} and \eqref{proof:DY}), a similar argument as in Lemma \ref{lem:Ds close} gives $\sup_{u\ge 0}\mathbb{E}^\mathbb{Q}[\vert D_uZ_t - D_uY_t\vert^p] = \mathcal{O}(t^{\frac{p}{2}})$, which yields 
	\begin{equation}
	\sup_{s\ge 0}\mathbb{E}^\mathbb{Q}\left[\left\vert D_s\left(\frac{1}{\frac{1}{T}\int_0^T Z_t\,dt}\right) - TD_sF\right\vert^p\right]=\mathcal{O}(T^{\frac{p}{2}})\,.
	\end{equation}
	
	Observe that 
    \begin{align}
	&\left\vert\mathbb{E}^\mathbb{Q}\left[\Phi\left(\frac{1}{T}\int_0^T S_t\,dt\right)\frac{1}{T}\delta(v)F_Z\right] - \mathbb{E}^\mathbb{Q}\left[\Phi\left(\frac{1}{T}\int_0^T X_t\,dt\right)\frac{1}{T}\delta(u)TF\right]\right\vert \nonumber \\
	&\le \frac{1}{T}\mathbb{E}^\mathbb{Q}\left[\left\vert\Phi\left(\frac{1}{T}\int_0^T S_t\,dt\right) - \Phi\left(\frac{1}{T}\int_0^T X_t\,dt\right)\right\vert\vert\delta(v)\vert F_Z\right] \label{eq:delta_err1} \\
	&\quad + \frac{1}{T}\mathbb{E}^\mathbb{Q}\left[\left\vert\Phi\left(\frac{1}{T}\int_0^T X_t\,dt\right)\right\vert\vert\delta(v - u)\vert F_Z\right] \label{eq:delta_err2} \\
	&\quad + \frac{1}{T}\mathbb{E}^\mathbb{Q}\left[\left\vert\Phi\left(\frac{1}{T}\int_0^T X_t\,dt\right)\right\vert\vert\delta(u)\vert \vert F_Z - TF\vert\right]\,. \label{eq:delta_err3}
	\end{align}
	Since $\Phi$ is $\gamma$-H\"older continuous, we have $\mathbb{E}^\mathbb{Q}[\vert \Phi(\frac{1}{T}\int_0^T S_t\,dt) - \Phi(\frac{1}{T}\int_0^T X_t\,dt) \vert^p] = \mathcal{O}(T^{\gamma p})$. Using the bounds $\mathbb{E}^\mathbb{Q}[\vert \delta(v) \vert^q] = \mathcal{O}(T^{\frac{q}{2}})$ and $\mathbb{E}^\mathbb{Q}[\vert F_Z \vert^r] = \mathcal{O}(1)$, and applying H\"older's inequality, we obtain that the first term \eqref{eq:delta_err1} is bounded by $\frac{1}{T}\mathcal{O}(T^\gamma)\mathcal{O}(T^{\frac{1}{2}})\mathcal{O}(1) = \mathcal{O}(T^{\gamma-\frac{1}{2}})$. 
	For the second term \eqref{eq:delta_err2}, since the integrand difference is bounded by $\mathbb{E}^\mathbb{Q}[\vert v_t-u_t \vert^p] = \mathcal{O}(t^p)$, the Burkholder--Davis--Gundy inequality implies $\mathbb{E}^\mathbb{Q}[\vert \delta(v - u) \vert^p] = \mathcal{O}(T^{\frac{3p}{2}})$, leading to an error bound of $\frac{1}{T}\mathcal{O}(1)\mathcal{O}(T^{\frac{3}{2}})\mathcal{O}(1) = \mathcal{O}(T^{\frac{1}{2}})$. 
	Using $\mathbb{E}^\mathbb{Q}[\vert F_Z - TF \vert^p] = \mathcal{O}(T^p)$, we have that the third term \eqref{eq:delta_err3} is bounded by $\frac{1}{T}\mathcal{O}(1)\mathcal{O}(T^{\frac{1}{2}})\mathcal{O}(T) = \mathcal{O}(T^{\frac{1}{2}})$. 
	Since $\gamma \le 1$, the error of $\mathcal{O}(T^{\frac{1}{2}})$ is subsumed into $\mathcal{O}(T^{\gamma-\frac{1}{2}})$.
	
	Applying an analogous decomposition to the second expectation in Proposition \ref{prop:asian delta mal holder}, we have the 
    \begin{align}\label{eqn:aaa}
    &\quad\,\,\mathbb{E}^\mathbb{Q}\left[\Phi\left(\frac{1}{T}\int_0^T S_t\,dt\right)\int_0^T \frac{2{Z_u}^2}{\sigma(u,S_u)S_u}D_{u}\left(\frac{1}{\int_0^T Z_t\,dt}\right)\,du\right]\\
    &=
    \mathbb{E}^\mathbb{Q}\left[\Phi\left(\frac{1}{T}\int_0^T X_t\,dt\right)\int_0^T \frac{2{Y_s}^2}{\sigma(s,X_s)X_s}D_{s}\left(\frac{1}{\int_0^T Y_t\,dt}\right)\,ds\right]+\mathcal{O}(T^{\gamma-\frac{1}{2}})\,.
    \end{align}
    By combining \eqref{eq:delta_err3} and  \eqref{eqn:aaa}, we have
	\begin{equation}
	\Delta_A(T) = e^{-rT}\frac{\partial}{\partial S_0}\mathbb{E}^\mathbb{Q}\left[\Phi\left(\frac{1}{T}\int_0^T X_t\,dt\right)\right] + \mathcal{O}(T^{\gamma-\frac{1}{2}})\,.
	\end{equation}
	Combining this with Propositions \ref{prop:Asian delta app 3} and \ref{prop:Asian delta app 1}, we obtain the desired result  \eqref{eqn:delta} for $\Delta_A(T)$.
\end{proof}

The borderline is $\gamma=\frac{1}{2}$ in these formulas. If $\gamma<\frac{1}{2}\,,$ the estimates in Theorem \ref{thm:delta formula holder} are meaningless; however, for $\frac{1}{2}<\gamma\le 1\,,$ they provide us with the short-maturity estimate with the convergence rate $\gamma-\frac{1}{2}\le \frac{1}{2}.$

\begin{example}\label{ex: gamma delta holder}
This example demonstrates that the asymptotic convergence order $\mathcal{O}(T^{\gamma-\frac{1}{2}})$ in Theorem \ref{thm:delta formula holder} cannot be further improved.	
	Given any $K$ and $\frac{1}{2}<\gamma<1,$ define the payoff function $\Phi$ by
	\begin{equation}
	\Phi(x)=(x-K)_{+}^{\gamma}\,.
	\end{equation}
	If $S_0=K,$ then we obtain the asymptotic equation
	\begin{equation}
	\Delta_A(T)=\frac{M(\gamma+1)}{2(S_0\sigma_A(T))^{1-\gamma}}{T^{\frac{\gamma-1}{2}}}+\mathcal{O}(T^{\gamma-\frac{1}{2}})\,,
	\end{equation}
	where $M(\gamma+1):=\mathbb{E}^\mathbb{Q}[\vert Z\vert^{\gamma+1}]$ with a standard normal variable $Z.$
\end{example}

In Theorem \ref{thm:delta formula holder}, the convergence rate of the expectations
$$\mathbb{E}^\mathbb{Q}\left[\frac{\Phi(S_0+S_0\sigma_A(T)\sqrt{T}Z)}{S_0\sigma_A(T)\sqrt{T}}Z\right]\,\textnormal{ and }\;\mathbb{E}^\mathbb{Q}\left[\frac{\Phi(S_0+S_0\sigma_E(T)\sqrt{T}Z)}{S_0\sigma_E(T)\sqrt{T}}Z\right]$$
depends on the function $\Phi.$ As an extreme case, if $\Phi$ is constant, then these expectations are zero, which means that the order is smaller than $\mathcal{O}(T^{\gamma-\frac{1}{2}}).$ Consider the ATM (i.e., $S_0=K$) option with  $\Phi(x)=(x-K)_+^{\gamma}$ and $\Phi(x)=(K-x)_+^{\gamma}$. If $\frac{1}{2}<\gamma<1,$  
these are   of the order greater than $\mathcal{O}(T^{\gamma-\frac{1}{2}})$ as presented in Example \ref{ex: gamma delta holder}.

The following corollary states that if either $\Delta_E(T)$ or $\Delta_A(T)$ converges as $T \to 0$, then the other also converges, and the two limits are equal.
This is directly obtained by
\begin{equation}
\begin{aligned}
\lim_{T\rightarrow{0}}\Delta_A(T)
&=	\lim_{T\rightarrow{0}}\mathbb{E}^\mathbb{Q}\left[\frac{\Phi(S_0+S_0\sigma_A(T)\sqrt{T}Z)}{S_0\sigma_A(T)\sqrt{T}}Z\right] =\lim_{\epsilon\rightarrow{0}}\mathbb{E}^\mathbb{Q}\left[\frac{\Phi(S_0+\epsilon Z)}{\epsilon}Z\right] \\
&=	\lim_{T\rightarrow{0}}\mathbb{E}^\mathbb{Q}\left[\frac{\Phi(S_0+S_0\sigma_E(T)\sqrt{T}Z)}{S_0\sigma_E(T)\sqrt{T}}Z\right]=\lim_{T\rightarrow{0}}\Delta_E(T)
\end{aligned}
\end{equation}
\begin{corollary} Under Assumptions \ref{classical assumption}
	and \ref{Holder payoff assumption}, if $\Delta_E(T)$ converges as $T\rightarrow{0}$, then $\Delta_A(T)$ also converges and vice versa. Moreover, $\lim_{T\rightarrow{0}}\Delta_A(T)
	=\lim_{T\rightarrow{0}}\Delta_E(T)$.
\end{corollary}

If $\Phi$ is Lipschitz (that is, $\gamma=1$), then we have the following corollary.
It is directly obtained by using 	$\mathbb{E}^\mathbb{Q} [\Phi (S_0+S_0\sigma_A(T)\sqrt{T}Z ) ]=\Phi(S_0)+\mathcal{O}(\sqrt{T})$ and $\mathbb{E}^\mathbb{Q} [\Phi (S_0+S_0\sigma_E(T)\sqrt{T}Z ) ]=\Phi(S_0)+\mathcal{O}(\sqrt{T})\,.$

\begin{corollary}\label{cor:asian delta limit}
	Under Assumptions \ref{classical assumption}
	and \ref{Holder payoff assumption}, if $\gamma=1,$ then
	\begin{align}
	\lim_{T\rightarrow{0}}\Delta_A(T)
	=\frac{D\Phi(S_0+)+D\Phi(S_0-)}{2}
	\end{align}
	where $D\Phi(S_0+)$ and  $D\Phi(S_0-)$ are 
	the right  and   left derivatives, respectively.
\end{corollary}

\begin{example}\label{ex:call,put delta}
	Let $\Delta_A^{\textnormal{call}}$ and $\Delta_A^{\textnormal{put}}$
	be the Asian call and put delta value with  the strike $K,$ i.e., the payoff functions are $\Phi(x)=(x-K)_+$ and $\Phi(x)=(K-x)_+,$ respectively.
	Then,  
	\begin{align}
	\Delta_A^{\textnormal{call}}(T)=\begin{cases}
	0+\mathcal{O}(\sqrt{T}), & \mbox{if  }S_0<K\,, \\
	\frac{1}{2}+\mathcal{O}(\sqrt{T}), & \mbox{if  }S_0=K\,, \\
	1+\mathcal{O}(\sqrt{T}), & \mbox{if  }S_0>K\,,
	\end{cases}\quad
	\Delta_A^{\textnormal{put}}(T)=\begin{cases}
	-1+\mathcal{O}(\sqrt{T}), & \mbox{if  }S_0<K\,, \\
	\frac{1}{2}+\mathcal{O}(\sqrt{T}), & \mbox{if  }S_0=K\,, \\
	0+\mathcal{O}(\sqrt{T}), & \mbox{if  }S_0>K\,.
	\end{cases}
	\end{align}
\end{example}
\begin{example}\label{ex:convex payoff}
	Given any $K\,,\delta>0$, and $1\le \gamma < 2\,,$ define the payoff function $\Phi$ by
	\begin{align}
	\Phi(x)=(x-K)^{\gamma}\mathbbm{1}_{\{K\le x<K+\delta\}}+\delta^{\gamma}\mathbbm{1}_{\{K+\delta\le x\}}\,.
	\end{align}
	Suppose that $S_0=K\,.$ Then, we get the following asymptotic equation.
	\begin{equation}
	\Delta_A(T)=\frac{1}{2}(S_0\sigma_A(T))^{\gamma -1}\,M(\gamma+1)T^{\frac{\gamma -1}{2}}+\mathcal{O}(\sqrt{T})\,,
	\end{equation}
	where $M(\gamma+1):=\mathbb{E}^\mathbb{Q}[\vert Z\vert^{\gamma+1}]$ with a standard normal variable $Z\,.$ In this example, the leading order of $\Delta_A(T)$ is $T^{\frac{\gamma-1}{2}}\,$ as $T\rightarrow{0}\,.$
\end{example}

In Theorem \ref{thm:delta formula holder}, the convergence rate of the expectation 
$$\mathbb{E}^\mathbb{Q}\left[\frac{\Phi(S_0+S_0\sigma_A(T)\sqrt{T}Z)}{S_0\sigma_A(T)\sqrt{T}}Z\right]$$ 
depends on the function $\Phi.$ As an extreme case, if $\Phi$ is constant, then this expectation is zero, which means that the order is smaller than $\mathcal{O}(T^{\gamma-\frac{1}{2}}).$
If $\gamma=1$ and
$D\Phi(S_0+)+D\Phi(S_0-)\neq 0,$
then the convergence rate 		
is greater than $\mathcal{O}(T^{\gamma-\frac{1}{2}})$ as presented in Corollary \ref{cor:asian delta limit}.   
In particular, for the ATM (i.e., $S_0=K$) call and put options with  $\Phi(x)=(x-K)_+$ and $\Phi(K-x)_+$, we have
$$\mathbb{E}^\mathbb{Q}\left[\frac{\Phi(S_0+S_0\sigma_A(T)\sqrt{T}Z)}{S_0\sigma_A(T)\sqrt{T}}Z\right]=\frac{1}{2} \,.$$

\section{Numerical tests}\label{sec:numerical}

\subsection{Black-Scholes model}

In this section, we conduct a numerical test for Asian call option prices for $\Phi(x)=(x-K)_{+}$ under the Black-Scholes model
$$dS_t=(r-q)S_t\,dt+\sigma S_t\,dW_t\,,\;S_0>0\,.$$
Table \ref{table: call_BS} compares  numerical results obtained from 
our method (Asymptotics),
the Monte Carlo simulation (MC) and the method proposed by \cite{PirjolDan2016SMAO} (PZ).
This test assumes the model parameters are as follows:
$$r=q=0\,,\;S_0=100\,,\;\sigma=0.3\,.$$
Table \ref{table: call_BS_error}   presents  the relative errors of our method and the PZ approach compared to the MC results.

Our method demonstrates strong performance for short-maturity ATM options, producing results similar to those of the MC and PZ methods. 
For example, the relative error of our method is approximately 0.059\% in the ATM case with maturity $T = 0.5$.
However, it struggles with long-maturity and OTM options, with particularly significant errors observed in deep OTM cases.
For the OTM case with strike $K=130$ and maturity $T = 0.5$, the relative error rises to approximately $68.809\%$.
Therefore, while our method is effective for short-maturity ATM options, it is not suitable for other option types, particularly those with longer maturities and that are deeply OTM.

This phenomenon is related to the rate of convergence. The convergence order of the error for call options is $\mathcal{O}(T)$ as presented in Theorem \ref{thm:option price holder}. In contrast, the convergence order of the PZ method is
$e^{-\frac{\mathcal{I}}{T}+o(\frac{1}{T})}$
where $\mathcal{I}$ is the associated rate function. Since the approximation error of order $\mathcal{O}(T)$ is much larger than the exponentially decaying order of the PZ method, our method performs poorly for long-maturity and deep OTM options  relative to the PZ method.


\begin{table*}
	\centering
	\scalebox{0.8}{
		\begin{tabular}{c c c c c c c c c c c c}
			\toprule
			\midrule
			\multirow{2}[4]{*}{value of $K$} & \multicolumn{2}{c}{$T=0.5$} & \multicolumn{4}{c}{$T=1$} & \multicolumn{2}{c}{$T=2$}\\ 
			\cmidrule(rl){2-4}\cmidrule(rl){5-7}\cmidrule(rl){8-10}
			& Asymptotics & MC & PZ & Asymptotics & MC & PZ & Asymptotics & MC & PZ \\
			\cmidrule(r){1-1}\cmidrule(l){2-4}\cmidrule(l){5-7}\cmidrule(l){8-10}
			\multicolumn{1}{l}{$K=100$}& 4.8842 & 4.8871 &  4.8830 & 6.9020 &6.9037 & 6.9013 & 9.7661 & 9.7417 &9.7477  \\
			\multicolumn{1}{l}{$K=105$}& 2.7980 & 2
   9205 &2.9188  & 4.7009 &4.8848 & 4.8847 & 7.5055 & 7.7268 & 7.7382\\
			\multicolumn{1}{l}{$K=110$}& 1.4322& 1.6372  & 1.6388 & 3.0413& 3.3689 & 3.3715 & 5.5800 & 6.0737 & 6.0826\\
			\multicolumn{1}{l}{$K=115$} & 0.6551& 0.8650  & 0.8671 & 1.8598 &2.2698  & 2.2745 & 4.0487& 4.7370 & 4.7505\\
			\multicolumn{1}{l}{$K=120$}& 0.2659 & 0.4336  & 0.4351 &1.0701& 1.4980 & 1.5033  & 2.8727& 3.6692 & 3.6835\\
			\multicolumn{1}{l}{$K=125$}& 0.0941& 0.2075 & 0.2081 & 0.5781& 0.9715 & 0.9758 &1.9662&2.8254  & 2.8414\\
			\multicolumn{1}{l}{$K=130$}& 0.0296 & 0.0949 & 0.0953 & 0.2920 & 0.6201 & 0.6234 & 1.3125 &2.1657 & 2.1790\\
			\midrule
			\bottomrule
	\end{tabular}}
	\caption{Asian call option prices for $\Phi(x)=(x-K)_{+}$ under the Black-Scholes model.} \label{table: call_BS}
\end{table*}

\begin{table*}
	\centering
	\scalebox{0.8}{
		\begin{tabular}{c c c c c c c c c c }
			\toprule
			\midrule
			\multirow{2}[4]{*}{value of $K$} & \multicolumn{2}{c}{$T=0.5$} & \multicolumn{2}{c}{$T=1$} & \multicolumn{2}{c}{$T=2$}\\ 
			\cmidrule(rl){2-3}\cmidrule(rl){4-5}\cmidrule(rl){6-7}
			& Asymptotics &   PZ & Asymptotics &   PZ & Asymptotics &  PZ \\
			\cmidrule(r){1-1}\cmidrule(l){2-4}\cmidrule(l){5-7}\cmidrule(l){8-10}
			\multicolumn{1}{l}{$K=100$}& 0.059\% &   0.084\% & 0.025\% &0.035\% & 0.250\% & 0.062\%  \\
			\multicolumn{1}{l}{$K=105$}& 4.194\% & 0.058\% &3.768\%	  & 0.002\% &2.864\% & 0.148\% \\
			\multicolumn{1}{l}{$K=110$}& 12.521\%& 0.098\%  & 9.724\% & 0.077\% & 8.128\% & 0.147\%
			  \\
			\multicolumn{1}{l}{$K=115$} & 24.266\% &0.248\%  & 18.063\% & 0.207\% &14.530\% &0.285\%  \\
			\multicolumn{1}{l}{$K=120$}& 38.676\% & 0.346\%  & 28.565\% &0.354\%& 21.708\% & 0.390\%  \\
			\multicolumn{1}{l}{$K=125$}& 54.651\%& 0.289\% & 40.494\% & 0.443\%& 30.410\% & 0.566\%  \\
			\multicolumn{1}{l}{$K=130$}& 68.809\% &0.421\% &52.911\% & 0.532\% & 39.396\% & 0.614\% \\
			\midrule
			\bottomrule
	\end{tabular}}
	\caption{Relative errors compared to the MC method presented in Table \ref{table: call_BS}}\label{table: call_BS_error}
\end{table*}

\subsection{Local volatility models}
\label{sec:LV}

In this section, we present numerical experiments for the asymptotic formulas in Theorems \ref{thm:option price holder} and \ref{thm:delta formula holder}. 
We consider two local volatility models: the CEV model and the quadratic model. 
For the volatility functions \(v\) given in \eqref{eqn:cev} (the CEV model) and \eqref{eqn:quad} (the quadratic model), the standard local volatility model
\begin{align}
	dS_t=(r-q)S_t\,dt+v(t,S_t)S_t\,dW_t
\end{align}
does not satisfy Assumption \ref{classical assumption}. 
To overcome this issue, we consider truncated versions of these models.
For each \(\ell\in\mathbb{N}\), let \(\eta_\ell:\mathbb{R}\to\mathbb{R}\) be a positive twice continuously differentiable function satisfying
$\eta_\ell(x)=x$ for $\frac{1}{\ell}\le x\le \ell$, $\eta_\ell(x)=\ell+1$ for $x>\ell+1,$ 
 $\eta_\ell(x)=\frac{1}{2\ell}$ for $x<\frac{1}{2\ell}$ and $\frac{1}{2\ell}\le \eta_\ell(x)\leq \ell+1$ for $x\in \mathbb{R}$.
Define the truncated volatility function \(\sigma_\ell\) with truncation level \(\ell\) by
\begin{equation}
\label{eqn:trun}
\sigma_\ell
=
\eta_\ell \circ v
:
[0,T]\times\mathbb{R}
\to
\mathbb{R}\,.
\end{equation}
Then the truncated local volatility model
\begin{align}
	dS_t=(r-q)S_t\,dt+\sigma_\ell(t,S_t)S_t\,dW_t
\end{align}
satisfies Assumption \ref{classical assumption}.
Moreover, this truncated model provides an accurate approximation to the original model when the truncation level \(\ell\) is sufficiently large since \(\sigma_\ell\) remains close to \(v\).
A similar approach was adopted by \cite{SACHSEKKEHARDW2014RMFT} to address the mismatch between technical assumptions and CEV model for their estimations of implied volatility.

\begin{enumerate}
	\item \textbf{Extended constant elasticity of variance (CEV) model:} The volatility function is given as $\sigma_\ell=\eta_\ell \circ v$ for
\begin{equation}
    \label{eqn:cev}
    v(t,x)=e^{-\lambda t}\xi x^{\theta-1}\,.
\end{equation} 
Our choices of parameters are  $r=q=0$, $S_0=100$,  $\lambda=1$, $\xi=0.2$, $\theta=0.5$ and $\ell=10^9.$ \cite{GatheralJim2012AOIV,SACHSEKKEHARDW2014RMFT} studied this CEV model without truncation. The same set of parameters except for $\lambda=0$ was used to test estimations of implied volatility.
	\item \textbf{Quadratic model:}
The volatility function is given as $\sigma_\ell=\eta_\ell \circ v$ for
	\begin{align}     \label{eqn:quad}
	v(t,x)=\frac{e^{-\lambda t}\sigma}{x}\left(\psi x+(1-\psi) S_0+\frac{\eta}{2}\frac{(x-S_0)^2}{S_0}\right).
	\end{align}
Our choices of parameters are 
$r=q=0$, $S_0=100$,
$\lambda=1$, $\sigma=0.2$, $\psi=0.5$, $\eta=10$ and $\ell=10^9.$ With the same choice of $(r,q,S_0,\sigma,\psi,\eta)$  without truncation, \cite{AndersenLeif2011Opwq} investigated implied volatility smile for time-independent cases $(\lambda=0)$. Here, we put $\lambda=1$ to study time-dependent cases as proposed in \cite{GatheralJim2012AOIV}.
\end{enumerate}

From a Monte Carlo perspective, the truncated-volatility model provides an effective approximation of the original CEV and quadratic models when \(\ell\) is sufficiently large. 
In practice, for large \(\ell\), the sample paths \((S_t)_{0\le t\le T}\) generated by the Euler scheme rarely reach the truncation boundaries \(\ell\) or \(1/\ell\). Consequently, the sample paths produced by the original and truncated models are practically indistinguishable.
Although the paths are not theoretically identical with probability one, since there remains a positive probability of hitting the truncation boundaries, this probability becomes negligibly small as \(\ell\) increases. Moreover, in the short-maturity regime considered in this paper, it is extremely unlikely that a sample path reaches the truncation boundaries before maturity.

\subsubsection{Powers of call options}\label{subsection: call option power}

\begin{table*}
	\centering
	\scalebox{0.8}{
	\begin{tabular}{c c c c c c c c c}
		\toprule
		\midrule
		\multirow{2}[4]{*}{value of $\gamma$} & \multicolumn{2}{c}{$T=1/365$} & \multicolumn{2}{c}{$T=1/10$} & \multicolumn{2}{c}{$T=1$}\\ 
		\cmidrule(rl){2-3}\cmidrule(rl){4-5}\cmidrule(rl){6-7}
		& MC & Asymptotics & MC & Asymptotics & MC & Asymptotics\\
		\cmidrule(r){1-1}\cmidrule(l){2-3}\cmidrule(l){4-5}\cmidrule(l){6-7}
		\multicolumn{1}{l}{$\gamma=0.1$}& 0.35586 & 0.35647 & 0.42815 & 0.4257
		& 0.46778 & 0.46856\\
		\multicolumn{1}{l}{$\gamma=0.6$}& 0.074969 & 0.075051 & 0.21755 & 0.21767 & 0.38484 & 0.38708 \\
		\multicolumn{1}{l}{$\gamma=1$}& 0.02389 & 0.024095 & 0.14123 & 0.14213 & 0.36948 & 0.37097 \\
		\multicolumn{1}{l}{$\gamma=1.5$}& 0.0063799 & 0.006383 & 0.091161 & 0.091443 & 0.38319 & 0.38559\\
		\multicolumn{1}{l}{$\gamma=1.9$}& 0.0023342 & 0.0023313 & 0.067727 & 0.067922 & 0.42306 & 0.42038\\
		\midrule
		\bottomrule
	\end{tabular}}
	\caption{ATM Asian option prices for $\Phi(x)=(x-K)^{\gamma}$ under the CEV model.}
	\label{table: cev price}
\end{table*}

\begin{table*}
	\centering
	\scalebox{0.8}{
	\begin{tabular}{c c c c c c c c c}
		\toprule
		\midrule
		\multirow{2}[4]{*}{value of $\gamma$} & \multicolumn{2}{c}{$T=1/365$} & \multicolumn{2}{c}{$T=1/10$} & \multicolumn{2}{c}{$T=1$}\\ 
		\cmidrule(rl){2-3}\cmidrule(rl){4-5}\cmidrule(rl){6-7}
		& MC & Asymptotics & MC & Asymptotics & MC & Asymptotics\\
		\cmidrule(r){1-1}\cmidrule(l){2-3}\cmidrule(l){4-5}\cmidrule(l){6-7}
		\multicolumn{1}{l}{$\gamma=0.1$}& 5.0116 & 5.0278 & 1.0197 & 1.0179 & 0.42619 & 0.42926\\
		\multicolumn{1}{l}{$\gamma=0.6$}& 1.3541 & 1.3547 & 0.67115 & 0.66609 & 0.44313 & 0.45381 \\
		\multicolumn{1}{l}{$\gamma=1$}& 0.4939 & 0.5 & 0.49631 & 0.5 & 0.50188 & 0.5 \\
		\multicolumn{1}{l}{$\gamma=1.5$}& 0.1518 & 0.15154 & 0.36411 & 0.36806 & 0.59254 & 0.59462\\
		\multicolumn{1}{l}{$\gamma=1.9$}& 0.060072 & 0.060393 & 0.2991 & 0.2983 & 0.71432 & 0.70736\\
		\midrule
		\bottomrule
	\end{tabular}}
	\caption{ATM Asian option deltas for $\Phi(x)=(x-K)^{\gamma}$ under the CEV model.}
	\label{table: cev delta}
\end{table*}

\begin{table*}
	\centering
	\scalebox{0.8}{
	\begin{tabular}{c c c c c c c c c}
		\toprule
		\midrule
		\multirow{2}[4]{*}{value of $\gamma$} & \multicolumn{2}{c}{$T=1/365$} & \multicolumn{2}{c}{$T=1/10$} & \multicolumn{2}{c}{$T=1$}\\ 
		\cmidrule(rl){2-3}\cmidrule(rl){4-5}\cmidrule(rl){6-7}
		& MC & Asymptotics & MC & Asymptotics & MC & Asymptotics\\
		\cmidrule(r){1-1}\cmidrule(l){2-3}\cmidrule(l){4-5}\cmidrule(l){6-7}
		\multicolumn{1}{l}{$\gamma=0.1$}& 0.44697 & 0.44877 & 0.52977 & 0.53592 & 0.57678 & 0.58988\\
		\multicolumn{1}{l}{$\gamma=0.6$}& 0.29916 & 0.29878 & 0.8577 & 0.86658 & 1.5252 & 1.541 \\
		\multicolumn{1}{l}{$\gamma=1$}& 0.23978 & 0.24095 & 1.428 & 1.4213 & 3.7326 & 3.7097\\
		\multicolumn{1}{l}{$\gamma=1.5$}& 0.20237 & 0.20185 & 2.9437 & 2.8917 & 12.7076 & 12.1934\\
		\multicolumn{1}{l}{$\gamma=1.9$}& 0.18516 & 0.18518 & 5.487 & 5.3952 & 36.3354 & 33.3923\\
		\midrule
		\bottomrule
	\end{tabular}}
	\caption{ATM Asian option prices for $\Phi(x)=(x-K)^{\gamma}$ under the quadratic model.}
	\label{table: quad price}
\end{table*}

\begin{table*}
	\centering
	\scalebox{0.8}{
	\begin{tabular}{c c c c c c c c c}
		\toprule
		\midrule
		\multirow{2}[4]{*}{value of $\gamma$} & \multicolumn{2}{c}{$T=1/365$} & \multicolumn{2}{c}{$T=1/10$} & \multicolumn{2}{c}{$T=1$}\\ 
		\cmidrule(rl){2-3}\cmidrule(rl){4-5}\cmidrule(rl){6-7}
		& MC & Asymptotics & MC & Asymptotics & MC & Asymptotics\\
		\cmidrule(r){1-1}\cmidrule(l){2-3}\cmidrule(l){4-5}\cmidrule(l){6-7}
		\multicolumn{1}{l}{$\gamma=0.1$}& 0.63037 & 0.63296 & 0.12679 & 0.12815 & 0.05282 & 0.05404\\
		\multicolumn{1}{l}{$\gamma=0.6$}& 0.53762 & 0.5393 & 0.26232 & 0.26517 & 0.1791 & 0.18066 \\
		\multicolumn{1}{l}{$\gamma=1$}& 0.49713 & 0.5 & 0.50838 & 0.5 & 0.51115 & 0.5\\
		\multicolumn{1}{l}{$\gamma=1.5$}& 0.48055 & 0.47923 & 1.1917 & 1.1639 & 1.9573 & 1.8804\\
		\multicolumn{1}{l}{$\gamma=1.9$}& 0.47943 & 0.47972 & 2.4184 & 2.3695 & 6.2038 & 5.6188\\
		\midrule
		\bottomrule
	\end{tabular}}
	\caption{ATM Asian option deltas for $\Phi(x)=(x-K)^{\gamma}$ under the quadratic model.}
	\label{table: quad delta}
\end{table*}

Under these two models, we will first compare asymptotic formulas for ATM $(K=S_0)$ Asian option prices and deltas having the payoff $\Phi(x)=(x-K)_{+}^{\gamma}$ presented in Examples \ref{ex: gamma price lipschitz}, \ref{ex:convex payoff}, \ref{ex: gamma price holder} and \ref{ex: gamma delta holder} with results from the Monte Carlo simulation. While these options are not traded in a real market, numerical tests could show how accurate our asymptotic formulas are.\\
\indent In Tables \ref{table: cev price}, \ref{table: cev delta}, \ref{table: quad price} and \ref{table: quad delta}, numerical results from asymptotic formulas presented in this paper are given in ``Asymptotics" columns. Results from the Monte Carlo simulation are given in ``MC" columns. We consider $M=10^5$ paths and $N=10^3$ time steps during simulations. For simulating deltas, we use the Malliavin representation in Proposition \ref{prop:asian delta mal holder}.

Tables \ref{table: cev price}, \ref{table: cev delta}, \ref{table: quad price} and \ref{table: quad delta} show that asymptotic formulas are being more accurate as the maturity $T$ gets shorter. 
This trend holds across a range of values for $\gamma$. 
For ATM options with $\gamma = 1$ and $T = 1/10$, the relative error for the price estimated by our method is approximately $0.637\%$ under the CEV model, while the relative error for the delta is about $0.743\%$.
Under the same conditions in the Quadratic model, the relative error for the option price is approximately $0.469\%$, and the relative error for the delta is about $1.648\%.$
It is also noteworthy that the asymptotic formulas tend to be more accurate under the CEV model than under the quadratic model.
These differences are expected, as the Quadratic model implies relatively higher local volatilities compared to the CEV model.

\subsubsection{Approximations of digital options}\label{subsection: logistic}

\begin{figure}
	\centering
	\includegraphics[scale = 0.4]{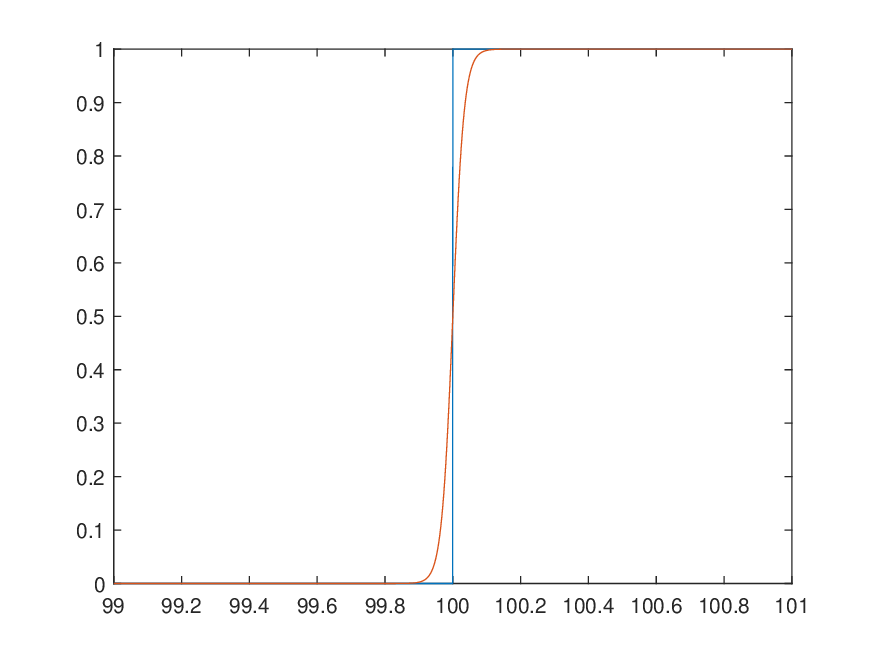}
	\caption{Plots of $\Phi_{\textnormal{binary}}$ and $\Phi_{\textnormal{logistic}}(\cdot:\kappa)$ with $K=100$ and $\kappa=60$.}
	\label{figure:logistic}
\end{figure}

\begin{table*}
	\centering
	\scalebox{0.8}{
	\begin{tabular}{c c c c}
		\toprule
		\midrule
		Value of $T$ & Digital option(MC) & Logistic option(MC)\\
		\cmidrule(r){1-1}\cmidrule(l){2-2}\cmidrule(l){3-3}
		\multicolumn{1}{l}{$T=1/10^6$}& 0.50081 & 0.50004\\
		\multicolumn{1}{l}{$T=1/10^4$}& 0.49976 & 0.49987 \\
		\multicolumn{1}{l}{$T=1/10^2$}& 0.49813 & 0.49827 \\
		\midrule
		\bottomrule
	\end{tabular}}
	\caption{Price approximation of digital options by logistic options under the CEV model.}
	\label{table:logistic price cev}
\end{table*}

\begin{table*}
	\centering
	\scalebox{0.8}{
	\begin{tabular}{c c c c}
		\toprule
		\midrule
		\multirow{2}[4]{*}{Value of $T$} & \multirow{2}[4]{*}{Digital option(MC)} & \multicolumn{2}{c}{Logistic option}\\
		\cmidrule(rl){3-4}&  & MC & Asymptotics \\
		\cmidrule(r){1-1}\cmidrule(l){2-2}\cmidrule(l){3-3}\cmidrule(l){4-4}
		\multicolumn{1}{l}{$T=1/10^6$}& 343.3846 & 14.4961 & 14.982\\
		\multicolumn{1}{l}{$T=1/10^4$}& 34.5661 & 13.5341 & 13.5279 \\
		\multicolumn{1}{l}{$T=1/10^2$}& 3.4577 & 3.3438 & 3.3521 \\
		\midrule
		\bottomrule
	\end{tabular}}
	\caption{Delta approximation of digital options by logistic options under the CEV model.}
	\label{table:logistic delta cev}
\end{table*}

\begin{table*}
	\centering
	\scalebox{0.8}{
	\begin{tabular}{c c c c}
		\toprule
		\midrule
		Value of $T$ & Digital option(MC) & Logistic option(MC)\\
		\cmidrule(r){1-1}\cmidrule(l){2-2}\cmidrule(l){3-3}
		\multicolumn{1}{l}{$T=1/10^8$}& 0.4996 & 0.50002 \\
		\multicolumn{1}{l}{$T=1/10^6$}& 0.49968 & 0.49996 \\
		\multicolumn{1}{l}{$T=1/10^4$}& 0.50112 & 0.50044 \\
		\midrule
		\bottomrule
	\end{tabular}}
	\caption{Price approximation of digital options by logistic options under the Quadratic model.}
	\label{table:logistic price quad}
\end{table*}

\begin{table*}
	\centering
	\scalebox{0.8}{
	\begin{tabular}{c c c c}
		\toprule
		\midrule
		\multirow{2}[4]{*}{Value of $T$} & \multirow{2}[4]{*}{Digital option(MC)} & \multicolumn{2}{c}{Logistic option}\\
		\cmidrule(rl){3-4}&  & MC & Asymptotics \\
		\cmidrule(r){1-1}\cmidrule(l){2-2}\cmidrule(l){3-3}\cmidrule(l){4-4}
		\multicolumn{1}{l}{$T=1/10^8$}& 344.38 & 14.6368 & 14.982 \\
		\multicolumn{1}{l}{$T=1/10^6$}& 34.5449 & 13.4539 & 13.5279 \\
		\multicolumn{1}{l}{$T=1/10^4$}& 3.4775 & 3.3679 & 3.3443 \\
		\midrule
		\bottomrule
	\end{tabular}}
	\caption{Delta Approximation of digital options by logistic options under the Quadratic model.}
	\label{table:logistic delta quad}
\end{table*}

Next, we consider the way to hedge \textit{Digital options} which are defined as having the terminal payoff $\Phi_{\textnormal{binary}}(x)=\mathbbm{1}_{\{x \ge K\}}$. The main technical issue regarding digital options is that their delta values are being unrealistically large at small $T$, especially when a spot price $S_0$ is close to a strike $K$.\\
\indent One way to deviate this problem is to replace digital options by more manageable options for hedging purposes. Consider options having logistic functions $\Phi_{\textnormal{logistic}}(x;\kappa)=1/(1+e^{-\kappa(x-K)})$, $\kappa>0$ as the terminal payoff. (Say these options as \textit{Logistic options}.) Logistic functions have following properties:
\begin{enumerate}
	\item They are Lipschitz continuous and $\Phi_{\textnormal{logistic}}(x;\kappa)\rightarrow(1/2)\mathbbm{1}_{\{x=K\}}+\mathbbm{1}_{\{x>K\}}$ as $\kappa\rightarrow\infty$.
	\item $0<\Phi'_{\textnormal{logistic}}(\cdot;\kappa)\le \kappa/4$ and $\Phi'_{\textnormal{logistic}}(K,\kappa)=\kappa/4$.
\end{enumerate}
These properties indicate two advantages of replacing binary options by logistic options for hedging purposes. First, prices of logistic options are close to prices of binary options for sufficiently large $\kappa$. Second, hedging logistic options is feasible in the sense that delta values do not explode at small $T$. This feature makes hedging logistic options in replace of binary options attractive. Corollary \ref{cor:asian delta limit} implies that the upper bounded of delta values of logistic options is close to $\kappa/4$ when $K=S_0$. Therefore by controlling $\kappa$, practitioners can manage the cost of hedging logistic options.\\
\indent We perform numerical test for $\kappa=60$ and $K=S_0$ under the CEV model and the quadratic model. In Tables \ref{table:logistic price cev} and \ref{table:logistic price quad}, computations of digital option prices and logistic options prices by the Monte Carlo simulation are listed in \textit{Digital option(MC)} and \textit{Logistic option(MC)} columns, respectively. Test results justify that risk-neutral valuations of logistic options are similar to that of binary options. Similarly in Tables \ref{table:logistic delta cev} and \ref{table:logistic delta quad}, delta values of digital options and logistic options from the Monte Carlo method are respectively given in the first and second columns. The third columns contain estimated delta values of logistic options obtained from Theorem \ref{thm:delta formula holder}. As one can see, delta values of binary options explode at small $T$ where as those of logistic options seem to be bounded above by $\kappa/4=15$. Also, estimations from Theorem \ref{thm:delta formula holder} and the Monte Carlo simulation are in good agreement. Therefore, estimations of delta values of logistic options by Theorem \ref{thm:delta formula holder} provide us with the new way to hedge digital options.

\section{Special case: Approximation for call and put options}
\label{sec:Special case}
This section   considers only the Asian call option, i.e., $\Phi(x)=(x-K)_{+}$, and the Asian put option, i.e., $\Phi(x)=(K-x)_{+}\,.$ The meanings of the following notations are self-explanatory:
\begin{equation}
P^{\textnormal{call}}_A(T)\,,\,\, P^{\textnormal{put}}_A(T)\,,\,\, \Delta^{\textnormal{call}}_A(T)\,,\,\, \Delta^{\textnormal{put}}_A(T)\,.
\end{equation}
The short-maturity behaviors of these four quantities have already been analyzed in Examples \ref{ex:call,put price} and \ref{ex:call,put delta}.
We apply the large deviation principle to gain further insight into their asymptotic properties, following an approach similar to that in \cite{PirjolDan2016SMAO} and \cite{shoshi2025some}.
See Appendix \ref{proof:otm delta} for the proofs of the following theorem and corollary.

\begin{theorem}\label{thm:otm delta}
Let Assumptions \ref{classical assumption} hold. 	Suppose that  the volatility function  $\sigma(t,x)$ is independent of $t$, that is, $\sigma(t,x)\equiv\sigma(x)$.
Then the following statements hold for the rate function $\mathcal{I}$   defined as
\begin{equation}\label{eq:rate function}
\mathcal{I}(x, y):=\underset{\substack{\int_0^1e^{g(t)}\,dt=x, \\ g(0)=\log y,\, g\in\mathcal{AC}[0,1]}}{\inf}\frac{1}{2}\int_0^1\,\left(\frac{g'(t)}{\sigma(e^{g(t)})}\right)^2dt
\end{equation}
for $x,y>0$
where $\mathcal{AC}[0,1]$ is the space of absolutely continuous functions on $[0,1]$.
	
	\begin{enumerate}\renewcommand{\labelenumi}{(\roman{enumi})}
			\renewcommand{\theenumi}{\roman{enumi}} 
		\item For an OTM Asian call option, i.e., $K>S_0$,
		\begin{equation}\label{eq:otm cdelta}
		\lim_{T\rightarrow0}T\log(\Delta_A^{\textnormal{call}}(T))=-\mathcal{I}(K,S_0)\,.
		\end{equation}
		\item For an OTM Asian put option, i.e., $S_0>K$,
		\begin{equation}\label{eq:otm pdelta}
		\lim_{T\rightarrow0}T\log(-\Delta_A^{\textnormal{put}}(T))=-\mathcal{I}(K,S_0)\,.
		\end{equation}
	\end{enumerate}
\end{theorem}

\begin{corollary}\label{cor:itm delta}
Let Assumptions \ref{classical assumption} hold. 	Suppose that  the volatility function  $\sigma(t,x)$ is independent of $t$, that is, $\sigma(t,x)\equiv\sigma(x)$.
Then, the following asymptotic relations hold.
	\begin{enumerate}\renewcommand{\labelenumi}{(\roman{enumi})}
			\renewcommand{\theenumi}{\roman{enumi}} 
		\item For an ITM Asian call option, i.e., $S_0>K$,
		\begin{align}
		\Delta_A^{\textnormal{call}}(T)=1-\frac{1}{2}(r+q)T+\left(\frac{r^2+rq+q^2}{6}\right)T^2+\mathcal{O}(T^3)\,.
		\end{align}
		\item For an ITM Asian put option, i.e., $K>S_0$,
		\begin{align}
		\Delta_A^{\textnormal{put}}(T)=-1+\frac{1}{2}(r+q)T-\left(\frac{r^2+rq+q^2}{6}\right)T^2+\mathcal{O}(T^3)\,.
		\end{align}
	\end{enumerate}
\end{corollary}

Observe that Corollary \ref{cor:itm delta} extends the result presented in Example \ref{ex:call,put delta}. The short rate 
$r$
and dividend rate 
$q$ determine the asymptotic order beyond $\sqrt{T}.$
Corollary \ref{cor:itm delta} provides results for option delta values, which were not addressed in \cite{PirjolDan2018SOAO} or \cite{shoshi2025some}.

\section{Conclusion}\label{sec:conclusion}
This paper presents a short-maturity asymptotic analysis of Asian options with an arbitrary H\"older continuous payoff function under a local volatility model.
Our main focus is on the asymptotic behavior of the Asian option price and its delta.
Both quantities are expressed in terms of the Asian volatility
\begin{equation}
\sigma_{A}(T)=\sqrt{\frac{1}{T^3}\int_0^T\sigma^2(t,S_0)(T-t)^2\,dt}\,.
\end{equation}
For sufficiently small $T>0,$ we establish that
\begin{align}
&P_A(T)=\mathbb{E}^\mathbb{Q}[\Phi(S_0+S_0\sigma_A(T)\sqrt{T}Z)]+\mathcal{O}(T^{\gamma})\,, \\ &\Delta_A(T)=\mathbb{E}^\mathbb{Q}\left[\frac{\Phi(S_0+S_0\sigma_A(T)\sqrt{T}Z)}{S_0\sigma_A(T)\sqrt{T}}Z\right]+\mathcal{O}(T^{\gamma-\frac{1}{2}})\,,
\end{align}
where 
$Z$ is a standard normal random variable and 
$\gamma$ denotes the H\"older exponent of the payoff function $\Phi.$



For future research, it would be of interest to extend the present analysis from local volatility models to local-stochastic volatility (LSV) models for Asian options.
The Asian volatility,
used as an effective approximation in this paper, suggests that many of our results could be generalized to settings where the volatility dynamics themselves evolve stochastically.
Another promising direction is the study of Asian options under jump-diffusion models.
Recently, \cite{pirjol2024asymptotics} analyzed short-maturity Asian option prices within such models; however, to the best of our knowledge, the corresponding delta hedging problem has not yet been investigated.
Developing a rigorous delta hedging framework for Asian options in the presence of jumps would therefore represent a valuable contribution to the literature.
In a related study,  \cite{pirjol2025asian} derived short-maturity asymptotics for Asian option prices in LSV models using large deviations theory.
For out-of-the-money options, the asymptotics were characterized by a rate function expressed as a variational problem.
They also provided analytical expressions for the at-the-money  implied volatility level, skew, and convexity of Asian options under a general LSV framework.
 As another direction for future research, it would be interesting to study short-maturity VIX options and options on realized variance. There have recently been several developments in this area, and the methodology developed in this paper may be applicable to analyze these products in the at-the-money regime.

\section*{Acknowledgement(s)}

\noindent  Financial support from the Institute for Research in Finance and Economics of Seoul National University is gratefully acknowledged.

\section*{Disclosure statement}

\noindent No potential conflict of interest was reported by the authors.

\section*{Funding}

\noindent Hyungbin Park was supported by the National Research Foundation of
Korea (NRF) grants funded by the Ministry of Science and ICT (Nos. 2021R1C1C1011675, 2022R1A5A6000840, RS-2026-25488333).


\bibliographystyle{plain}
\bibliography{Asian2026}

\appendix

\section{Approximation scheme}
\label{sec:Approximation scheme}

We recall the six processes 
$X,Y,\tilde{X},\tilde{Y},\hat{X},\hat{Y}$
introduced at the beginning of Section~\ref{sec:Short maturity asymptotic for a sensitivity}.
\begin{lemma} \label{lem:x,y close}
Let Assumption \ref{classical assumption} hold. For any $p>0,$ there exists a positive constant $B_p$ depending only on $p$ such that the following inequalities hold.
	\begin{enumerate}\renewcommand{\labelenumi}{(\roman{enumi})}
			\renewcommand{\theenumi}{\roman{enumi}} 
		\item For $0\le t\le1,$  
		\begin{align} \label{eq:x close}
		\mathbb{E}^\mathbb{Q}[\vert X_t-\tilde{X}_t\vert^p] \le B_pt^p, \quad \mathbb{E}^\mathbb{Q}[\vert \tilde{X}_t-\hat{X}_t\vert^p] \le B_pt^p. 
		\end{align}
		\item  For $0\le t\le1,$  
		\begin{align} \label{eq:y close}
		\mathbb{E}^\mathbb{Q}[\vert Y_t-\tilde{Y}_t\vert^p] \le B_pt^p, \quad \mathbb{E}^\mathbb{Q}[\vert \tilde{Y}_t-\hat{Y}_t\vert^p] \le B_pt^p.
		\end{align}
	\item As $t\to 0,$ 
		\begin{align}\label{eqn:t}
		\mathbb{E}^\mathbb{Q}[\vert S_t-X_t\vert^p]=\mathcal{O}(t^p)\,, \quad \mathbb{E}^\mathbb{Q}[\vert Z_t-Y_t\vert^p]=\mathcal{O}(t^p)
		\end{align}
        where $Z$ is the unique solution to the SDE \eqref{eq:process z}.
	\end{enumerate}
\end{lemma}

\begin{proof}
We first prove the second inequality in \eqref{eq:x close} for $p \ge 2$. Once this is shown, the case $0 < p < 2$ follows from Jensen’s inequality, which yields
\begin{equation}\label{eqn:jensen}
\big(\mathbb{E}^{\mathbb{Q}}\big[|\tilde{X}_t-\hat{X}_t|^p \big]\big)^{\frac{2}{p}}
\le \mathbb{E}^{\mathbb{Q}}\big[|\tilde{X}_t-\hat{X}_t|^2 \big]
\le B_2 t^2
\end{equation}
for some constant $B_2>0$.	For $p\ge2$, observe that  
	\begin{align}
	\mathbb{E}^\mathbb{Q}[\vert \tilde{X}_t-\hat{X}_t \vert^p]
	&\le C_p\mathbb{E}^\mathbb{Q}\left[\left(\int_0^t\vert \sigma(s,S_0)\tilde{X}_s-\sigma(s,S_0)S_0 \vert^2 \,ds\right)^{\frac{p}{2}}\right]\\
	&\le C_p\overline{\sigma}^pt^{\frac{p}{2}-1}\int_0^t \mathbb{E}^\mathbb{Q}[ \vert \tilde{X}_s-S_0 \vert^p]\,ds \label{proof:x close 1.1}
	\end{align}
	for some constant $C_p>0$. For these inequalities, we have used the Burkholder--Davis--Gundy inequality, Assumption \ref{classical assumption}, and  Jensen's inequality. 
	Using Jensen's inequality and Theorem 3.4.3 of \cite{zhang2017backward}, it follows that for $t\le 1$,
	\begin{align}
	t^{\frac{p}{2}-1}\int_0^t \mathbb{E}^\mathbb{Q}[ \vert \tilde{X}_s-S_0 \vert^p]\,ds \le\,& 2^{p-1}\int_0^t\mathbb{E}^\mathbb{Q}[\vert \tilde{X}_s-\hat{X}_s \vert^p]\,ds  +2^{p-1}t^{\frac{p}{2}-1}\int_0^t\mathbb{E}^\mathbb{Q}[ \vert \hat{X}_s-S_0 \vert^p]\,ds \label{proof:x close 1.2}\\
	\le\,& 2^{p-1}\int_0^t\mathbb{E}^\mathbb{Q}[\vert \tilde{X}_s-\hat{X}_s \vert^p]\,ds +
	2^{p-1}\overline{\sigma}^{p}{S_0}^p\tilde{C}_p\frac{1}{\frac{p}{2}+1}t^p\label{proof:x close 1.3}
	\end{align}
	for some constant $\tilde{C}_p>0$.
	Hence, from \eqref{proof:x close 1.1} and \eqref{proof:x close 1.3}, we obtain
	\begin{align}
	\mathbb{E}^\mathbb{Q}[\vert \tilde{X}_t-\hat{X}_t \vert^p]
	\le f_p(t)+A_p\int_0^t\mathbb{E}^\mathbb{Q}[\vert \tilde{X}_s-\hat{X}_s \vert^p]\,ds\,,
	\end{align}
	where $f_p(t):=C_p\overline{\sigma}^p2^{p-1}\overline{\sigma}^{p}{S_0}^p\tilde{C}_p\frac{1}{\frac{p}{2}+1}t^p$ and $A_p:=C_p\overline{\sigma}^p2^{p-1}.$ Then, by  Gronwall's inequality, there is a constant $B_p>0$ such that $\mathbb{E}^\mathbb{Q}[\vert \tilde{X}_t-\hat{X}_t\vert^p] \le B_pt^p$  for all $0\le t\le 1$.  
	
	For the first inequality of \eqref{eq:x close}, we also provide the proof for $p\ge2$. By the Burkholder--Davis--Gundy inequality and Jensen's inequality, we have
	\begin{align}
	\mathbb{E}^\mathbb{Q}[\vert X_t-\tilde{X}_t \vert^p] \le C_pt^{\frac{p}{2}-1}\int_0^t\mathbb{E}^\mathbb{Q}[\vert \sigma(s,X_s)X_s-\sigma(s,S_0)\tilde{X}_s \vert^p]\,ds
	\end{align}
	for some constant $C_p>0$. It follows that
	\begin{align}
	&\quad t^{\frac{p}{2}-1}\int_0^t\mathbb{E}^\mathbb{Q}[\vert \sigma(s,X_s)X_s-\sigma(s,S_0)\tilde{X}_s \vert^p]\,ds \\
    &\le 3^{p-1}t^{\frac{p}{2}-1}\int_0^t\mathbb{E}^\mathbb{Q}[\vert \sigma(s,X_s)X_s-\sigma(s,\tilde{X}_s)\tilde{X}_s \vert^p]\,ds \\
	&+3^{p-1}t^{\frac{p}{2}-1}\int_0^t\mathbb{E}^\mathbb{Q}[\vert \sigma(s,\tilde{X}_s)\tilde{X}_s-\sigma(s,\hat{X}_s)\tilde{X}_s \vert^p]\,ds \\
	&+3^{p-1}t^{\frac{p}{2}-1}\int_0^t\mathbb{E}^\mathbb{Q}[\vert \sigma(s,\hat{X}_s)\tilde{X}_s-\sigma(s,S_0)\tilde{X}_s \vert^p]\,ds\,.
	\end{align}
	Observe   that for $0\le t\le1$,
	\begin{align}
	t^{\frac{p}{2}-1}\int_0^t\mathbb{E}^\mathbb{Q}[\vert \sigma(s,X_s)X_s-\sigma(s,\tilde{X}_s)\tilde{X}_s \vert^p]\,ds
	\le \alpha^p\int_0^t\mathbb{E}^\mathbb{Q}[\vert X_s-\tilde{X}_s \vert^p]\,ds\,.
	\end{align}
	By Assumption \ref{classical assumption},  H\"older's inequality and the second inequality of \eqref{eq:x close}, for $0\le t\le1$,
	\begin{equation} 
	\begin{aligned} \label{proof:x close 1.6}
	&\quad t^{\frac{p}{2}-1}\int_0^t\mathbb{E}^\mathbb{Q}[\vert \sigma(s,\tilde{X}_s)\tilde{X}_s-\sigma(s,\hat{X}_s)\tilde{X}_s \vert^p]\,ds \\
    &\le t^{\frac{p}{2}-1}\int_0^t(\mathbb{E}^\mathbb{Q}[\vert \sigma(s,\tilde{X}_s)-\sigma(s,\hat{X}_s) \vert^{2p}])^{\frac{1}{2}}(\mathbb{E}^\mathbb{Q}[\vert\tilde{X}_s\vert^{2p}])^{\frac{1}{2}} \,ds \\
	&\le t^{\frac{p}{2}-1}\alpha^p\int_0^t(\mathbb{E}^\mathbb{Q}[\vert \tilde{X}_s-\hat{X}_s\vert^{2p}])^{\frac{1}{2}}{S_0}^p e^{\frac{p(2p-1)\overline{\sigma}^2}{2}s}\,ds\\
	&\le \alpha^p(B_{2p})^{\frac{1}{2}}{S_0}^p t^{\frac{p}{2}-1}\int_0^t s^p e^{\frac{p(2p-1)\overline{\sigma}^2}{2}s}\,ds\\
	&\le \alpha^p(B_{2p})^{\frac{1}{2}}{S_0}^p t^{p}\int_0^1 e^{\frac{p(2p-1)\overline{\sigma}^2}{2}s}\,ds\,. 
	\end{aligned}	
	\end{equation}
Using $0\le t\le 1$ and $p\ge2$, we have
	\begin{align}
	t^{\frac{p}{2}-1}\int_0^t\mathbb{E}^\mathbb{Q}[\vert \sigma(s,\hat{X}_s)\tilde{X}_s-\sigma(s,S_0)\tilde{X}_s \vert^p]\,ds &\le \alpha^p\overline{\sigma}^p{S_0}^{2p}(\tilde{C}_{2p})^{\frac{1}{2}} t^{\frac{p}{2}-1}\int_0^t s^{\frac{p}{2}} e^{\frac{p(2p-1)\overline{\sigma}^2}{2}s}\,ds\,\label{proof:x close 1.7} \\
	&\le \alpha^p\overline{\sigma}^p{S_0}^{2p}(\tilde{C}_{2p})^{\frac{1}{2}}\frac{1}{\frac{p}{2}+1} t^{p} e^{\frac{p(2p-1)\overline{\sigma}^2}{2}\vee 0}\,.\label{proof:x close 1.8}
	\end{align}
 By combining the above three inequalities with Gronwall's inequality,  there exists a constant $B_p'>0$ such that $\mathbb{E}^\mathbb{Q}[\vert X_t-\tilde{X}_t \vert^p]\le B_p't^p$ for $t\le1.$
 
The proof of the second inequality in \eqref{eq:y close} is similar to that of the second inequality in \eqref{eq:x close} and is therefore omitted. We examine only the first inequality of \eqref{eq:y close}. 
	Observe that
	\begin{align}
	Y_t-\tilde{Y}_t=\int_0^t (\nu(s,X_s)Y_s-\nu(s,X_s)\tilde{Y}_s) \,ds &+\int_0^t  (\nu(s,X_s)\tilde{Y}_s-\nu(s,\hat{X}_s)\tilde{Y}_s)\,ds\\
	&+\int_0^t  (\nu(s,\hat{X}_s)\tilde{Y}_s-\nu(s,S_0)\tilde{Y}_s)\,ds\,.
	\end{align}
	Following the proof of \eqref{eq:x close}, we obtain the desired results.
	The asymptotics in 
	\eqref{eqn:t} also can be shown similarly.	
\end{proof}


The proof of the following lemma is standard and is therefore omitted.

\begin{lemma}\label{lem:dummy general}
Let $(W_t)_{t\ge0}$ be a Brownian motion on a   probability space $(\Omega, \mathcal{F}, \mathbb{Q})$.
Suppose $(\theta_t)_{t\ge0}$ is a process adapted to the Brownian filtration $(\mathcal{F}_t^W)_{t\ge 0}$ and is uniformly bounded.  Define a continuous martingale process $(M_t)_{t\ge0}$ as
	\begin{equation}
	M_t:=M_0e^{-\frac{1}{2}\int_0^t \theta_s^2\,ds+\int_0^t\theta_s\,dW_s}\,, \quad M_0>0\,.
	\end{equation}
	Then, for any $\xi\in\mathbb{R}$, the following three statements hold.
	\begin{enumerate}\renewcommand{\labelenumi}{(\roman{enumi})}
			\renewcommand{\theenumi}{\roman{enumi}} 
		\item $\begin{aligned}\lim_{T\rightarrow{0}}\mathbb{E}^\mathbb{Q}[M_T^{\xi}]=M_0^{\xi}\,.\end{aligned}$
		\item $\begin{aligned} \mathbb{E}^\mathbb{Q}\Big[\max_{0\le t\le T}M_t^{\xi}\Big]<\infty\,, \end{aligned}$ for any $T>0\,.$ Furthermore, $\begin{aligned} \lim_{T\rightarrow{0}}\mathbb{E}^\mathbb{Q}\Big[\max_{0\le t\le T}M_t^{\xi}\Big]=M_0^{\xi}\,. \end{aligned}$
		\item $\begin{aligned} \lim_{T\rightarrow{0}}\mathbb{E}^\mathbb{Q}\left[\left(\frac{1}{T}\int_0^T M_t\,dt\right)^\xi\right]=M_0^{\xi}\,. \end{aligned}$
	\end{enumerate}
\end{lemma}

We now present the short-time behavior of the four processes
$(X_t)_{t\ge0}$, $(\tilde{X}_t)_{t\ge0}$, $(Y_t)_{t\ge0}$, $(\tilde{Y}_t)_{t\ge0}$ in the following lemma. All moments of the random variables $X_T, \tilde{X}_T, Y_T, \tilde{Y}_T$ and their integrals over $[0,T]$ converge to constants as $T\rightarrow{0}.$ We formalize this observation in the following technical statement for later use.

\begin{lemma}\label{lem:dummy} Suppose that Assumption \ref{classical assumption} holds.  Consider the processes $X, \tilde{X}, Y, \tilde{Y}$ introduced in \eqref{eq:X_t},\eqref{eq:Y_t},\eqref{eq:tilde x,y}. 
For any $p_i\in\mathbb{R},$ $i\in\{1,2,3,4\}$, we define
$Z^{p_1,p_2,p_3,p_4}:=X^{p_1}\tilde{X}^{p_2}Y^{p_3}\tilde{Y}^{p_4}$. 
Then the following  properties hold.

	\begin{enumerate}\renewcommand{\labelenumi}{(\roman{enumi})}
			\renewcommand{\theenumi}{\roman{enumi}} 
		\item $\begin{aligned}\mathbb{E}^\mathbb{Q}\Big[\max_{0\le t\le T}Z_t^{p_1,p_2,p_3,p_4}\Big]<\infty\,\end{aligned}$ for any $T>0\,.$\,\, Furthermore, 
        \begin{equation}\begin{aligned}		\lim_{T\rightarrow{0}}\mathbb{E}^\mathbb{Q}\Big[\max_{0\le t\le T}Z_t^{p_1,p_2,p_3,p_4}\Big]=S_0^{p_1+p_2}\,.
		\end{aligned}            
        \end{equation}
		\item Moreover, for any $q_j\in\mathbb{R}\,,$ $j\in\{1,2,\cdots\,8\}\,,$
		\begin{align}
		&\lim_{T\rightarrow{0}}\mathbb{E}^\mathbb{Q}\Bigg[Z_T^{q_1,q_2,q_3,q_4}\left(\frac{1}{T}\int_0^T X_t\,dt\right)^{q_5}\left(\frac{1}{T}\int_0^T \tilde{X}_t\,dt\right)^{q_6}
		\left(\frac{1}{T}\int_0^TY_t\,dt\right)^{q_7}\left(\frac{1}{T}\int_0^T\tilde{Y}_t\,dt\right)^{q_8}\Bigg] \\
		&=S_0^{q_1+q_2+q_5+q_6}\,.\label{eq:dummy}
		\end{align}
	\end{enumerate}
\end{lemma}

\begin{proof}
 Observe that all processes $X, \tilde{X}, Y, \tilde{Y}$ satisfy the assumptions of Lemma~\ref{lem:dummy general}. Hence, from Lemma \ref{lem:dummy general} and  H\"older's inequality,
	we have $\mathbb{E}^\mathbb{Q}[\underset{0\le t\le T}{\max}Z_t^{p_1,p_2,p_3,p_4}]<\infty$  and $\limsup_{T\rightarrow{0}}\mathbb{E}^\mathbb{Q}[\underset{0\le t\le T}{\max}Z_t^{p_1,p_2,p_3,p_4}]\le S_0^{p_1+p_2}$. Moreover, by  Fatou's lemma,   $S_0^{p_1+p_2}\le \liminf_{T\rightarrow{0}}\mathbb{E}^\mathbb{Q}[\underset{0\le t\le T}{\max}Z_t^{p_1,p_2,p_3,p_4}].$ The inequality~\eqref{eq:dummy} can be proved in a similar manner.
\end{proof}

\section{Short-maturity limit of option prices}
\label{sec:Short maturity limit of an option price}




\begin{lemma}\label{lem:price drift 0}
	Under Assumptions \ref{classical assumption}
	and \ref{Holder payoff assumption}, as $T\rightarrow{0}$, we have
	\begin{enumerate}\renewcommand{\labelenumi}{(\roman{enumi})}
			\renewcommand{\theenumi}{\roman{enumi}} 
		\item $ \begin{aligned} P_A(T)=e^{-rT}\,\mathbb{E}^\mathbb{Q}\left[\Phi\left(\frac{1}{T}\int_0^T X_t\,dt\right)\right]+\mathcal{O}(T^\gamma)\,,\end{aligned}$
		\item $ \begin{aligned} P_E(T)=e^{-rT}\,\mathbb{E}^\mathbb{Q}[\Phi(X_T)]+\mathcal{O}(T^\gamma)\,. \end{aligned}$
	\end{enumerate}
\end{lemma}
\begin{proof} Observe that
	\begin{equation}
	\begin{aligned}
	\mathbb{E}^\mathbb{Q}[|S_t-X_t|^2]
	&=2	\mathbb{E}^\mathbb{Q}\left[\Big|\int_0^t(r-q)S_u\,du\Big|^2\right]+2	\mathbb{E}^\mathbb{Q}\left[\Big|\int_0^t\sigma(u,S_u)S_u-\sigma(u,X_u)X_u\,dW_u \Big|^2\right]\\
	&\le 2t	\mathbb{E}^\mathbb{Q}\left[ \int_0^t(r-q)^2S_u^2\,du \right]+2	\mathbb{E}^\mathbb{Q}\left[ \int_0^t(\sigma(u,S_u)S_u-\sigma(u,X_u)X_u)^2\,du\right]\\
	&\le 2t	\mathbb{E}^\mathbb{Q}  \int_0^t(r-q)^2\mathbb{E}^\mathbb{Q}[S_u^2]\,du +2\alpha^2	 \int_0^t\mathbb{E}^\mathbb{Q}[|S_u-X_u|^2]\,du
\end{aligned}
	\end{equation}
	where $\alpha$ is the Lipschitz constant in Assumption \ref{classical assumption}. 
	By Gronwall's Inequality, we have 
	$$	\mathbb{E}^\mathbb{Q}[|S_t-X_t|^2]\leq 2e^{2\alpha^2t}	 t\int_0^t(r-q)^2\mathbb{E}^\mathbb{Q}[S_u^2]\,du \leq ct^2$$
	for some positive constant $c$ independent of $t\in [0,T].$ 
	For $0<\gamma\le1,$
	$$	\mathbb{E}^\mathbb{Q}[|X_t-S_t|^\gamma]^{\frac{2}{\gamma}}\leq 	\mathbb{E}^\mathbb{Q}[|X_t-S_t|^2]\le ct^2\,.$$
	It follows that $\mathbb{E}^\mathbb{Q}[|X_t-S_t|^\gamma]\leq Ct^\gamma$ for some positive constant $C.$ 	
	From this inequality and Lemma \ref{lem:dummy}, we obtain the desired proof.
\end{proof}

As can be seen in \eqref{eq:S_t} and \eqref{eq:X_t}, the processes $S$ and $X$ have the same diffusion terms; however, the drift term of $X$ is zero.
Thus, this lemma implies that 
while estimating the Asian and European option prices,
the drift of the underlying stock becomes negligible at small $T>0$.  In \cite{PirjolDan2016SMAO, PirjolDan2019SMAO}, the rate function that governs the short-maturity behavior of the Asian call and put option was shown to be independent of the drift term.

The following is the proof of Theorem \ref{thm:option price holder}.

\begin{proof} 
	Choose $q>1$ such that $\gamma q>1\,.$ Then, by Lemma \ref{lem:price drift 0}, as $T\rightarrow{0}$, we have $$P_A(T)=\mathbb{E}^\mathbb{Q}\left[\Phi\left(\frac{1}{T}\int_0^T X_t\,dt\right)\right]+\mathcal{O}(T^{\gamma})\,.$$ 
By  Lemma \ref{lem:x,y close},   
\begin{align}
\quad\left\vert \mathbb{E}^\mathbb{Q}\left[\Phi\left(\frac{1}{T}\int_0^T X_t\,dt\right)\right]-\mathbb{E}^\mathbb{Q}\left[\Phi\left(\frac{1}{T}\int_0^T \hat{X}_t\,dt\right)\right]\right\vert^{\frac{1}{\gamma}}
&\le \beta \left\vert \mathbb{E}^\mathbb{Q}\left[\left|\frac{1}{T}\int_0^TX_t-\hat{X}_t\,dt\right|^{\gamma}\right]\right\vert^{\frac{1}{\gamma}} \\
&\le \beta  \mathbb{E}^\mathbb{Q}\left[\left|\frac{1}{T}\int_0^TX_t-\hat{X}_t\,dt\right|\right]\\
&\le   \frac{\beta}{T}\int_0^T\mathbb{E}^\mathbb{Q}[\vert X_t-\hat{X}_t\vert]\,dt\le C T\,
\end{align}
for some positive constant $C$.
This yields  	$$P_A(T)=\mathbb{E}^\mathbb{Q}\left[\Phi\left(\frac{1}{T}\int_0^T \hat{X}_t\,dt\right)\right]+\mathcal{O}(T^{\gamma})=\mathbb{E}^\mathbb{Q}[\Phi(S_0+S_0\sigma_A(T)\sqrt{T}Z)]+\mathcal{O}(T^{\gamma})\,.$$
For the last equality,  
	we have used
	\begin{align}
	\mathbb{E}^\mathbb{Q}\left[\Phi\left(\frac{1}{T}\int_0^T \hat{X}_t\,dt\right)\right]
	&=\mathbb{E}^\mathbb{Q}\left[\Phi\left(S_0+\frac{S_0}{T}\int_0^T \int_0^t\sigma(s,S_0)\,dW_s\,dt\right)\right] \\
	&=\mathbb{E}^\mathbb{Q}\left[\Phi\left(S_0+\frac{S_0}{T}\int_0^T \sigma(s,S_0)(T-s)\,dW_s\right)\right]\,,
	\end{align}
which is obtained from 	the stochastic Fubini theorem.  	
Similarly, using Lemma~\ref{lem:price drift 0}, we obtain $$	P_E(T)=\mathbb{E}^\mathbb{Q}[\Phi(S_0+S_0\sigma_E(T)\sqrt{T}Z)]+\mathcal{O}(T^{\gamma})\,.$$  
This completes the proof.	
\end{proof}

\section{Technical lemmas}


We recall the processes $(u_s)_{0\le s\le T}$, $(\hat{u}_s)_{0\le s\le T}$ and the random variables $F$, $\hat{F}$ defined in \eqref{eqn:uF} and \eqref{eqn:hat_uF}.
Furthermore, we define
$$ \tilde{u}_s:=\frac{2{\tilde{Y}_s}^2}{\sigma(s,\tilde{X}_s)\tilde{X}_s}\,, \; \tilde{F}:=\frac{1}{\int_0^T \tilde{Y}_t\,dt}\,.$$

\begin{lemma}\label{lem:u,TF app}
	Let Assumption \ref{classical assumption} hold. Then, for any $p>0,$ there exists a positive constant $D_p$ depending only on $p$ such that the following inequalities hold. 
	\begin{enumerate}\renewcommand{\labelenumi}{(\roman{enumi})}
			\renewcommand{\theenumi}{\roman{enumi}} 
		\item For $0\le t\le 1$,
		\begin{align}\label{eq:u app}
		\mathbb{E}^\mathbb{Q}[\left\vert u_t-\tilde{u}_t \right\vert^p] \le D_pt^p, \quad \mathbb{E}^\mathbb{Q}[\left\vert \tilde{u}_t-\hat{u}_t \right\vert^p]\le D_pt^p.
		\end{align}
		\item For $0\le T\le 1$,
		\begin{align}\label{eq:TF app}
		\mathbb{E}^\mathbb{Q}[\vert TF-T\tilde{F} \vert^p] \le D_pT^p, \quad \mathbb{E}^\mathbb{Q}[\vert T\tilde{F}-T\hat{F} \vert^p] \le D_pT^p.
		\end{align}
	\end{enumerate}
\end{lemma}

\begin{proof}
It suffices to establish the result for $p \ge 1$. The case $0<p<1$ follows by applying Jensen’s inequality in the same manner as in \eqref{eqn:jensen}. For the first inequality in \eqref{eq:u app}, observe that
	\begin{align}
	\vert u_t-\tilde{u}_t \vert^p &=\left\vert \frac{2(Y_t+\tilde{Y}_t)(Y_t-\tilde{Y}_t)}{\sigma(t,X_t)X_t}+\frac{2{\tilde{Y}_t}^2(\sigma(t,\tilde{X}_t)\tilde{X}_t-\sigma(t,X_t)X_t)}{\sigma(t,X_t)X_t\sigma(t,\tilde{X}_t)\tilde{X}_t} \right\vert^p\,.
	\end{align}
	From Lemmas \ref{lem:x,y close}, \ref{lem:dummy} and  H\"older's inequality, we obtain
	$\mathbb{E}^\mathbb{Q}[\vert u_t-\tilde{u}_t \vert^p]\le D_pt^p$.
	The second inequality of \eqref{eq:u app} can be proven as follows. We have
	\begin{align}
	&\mathbb{E}^\mathbb{Q}[\vert\tilde{u}_t-\hat{u}_t\vert^p] \\
	&\le 2^{p-1}\mathbb{E}^\mathbb{Q}\left[\left\vert \frac{2\tilde{Y}_t^2}{\sigma(t,\tilde{X}_t)\tilde{X}_t}\right\vert^p\mathbbm{1}_{\{\hat{X}_t<\frac{S_0}{2}\}}\right] 
	+2^{p-1}\mathbb{E}^\mathbb{Q}\left[\left\vert\frac{2\tilde{Y}_t^2}{\sigma(t,\tilde{X}_t)\tilde{X}_t}-\frac{2\hat{Y}_t^2}{\sigma(t,\hat{X}_t)\hat{X}_t}\right\vert^p\mathbbm{1}_{\{\hat{X}_t\ge \frac{S_0}{2}\}}\right].
	\end{align}
	Following the same argument as above, there exists a positive constant $D_p$ such that
	\begin{align}
	\mathbb{E}^\mathbb{Q}\left[\left\vert\frac{2\tilde{Y}_t^2}{\sigma(t,\tilde{X}_t)\tilde{X}_t}-\frac{2\hat{Y}_t^2}{\sigma(t,\hat{X}_t)\hat{X}_t}\right\vert^p\mathbbm{1}_{\{\hat{X}_t\ge \frac{S_0}{2}\}}\right]\le D_pt^p\,
	\end{align}
	for all $0\le t\le 1.$ From Assumption \ref{classical assumption} and  H\"older's inequality, it follows that
	\begin{align}
	\mathbb{E}^\mathbb{Q}\left[\left\vert \frac{2\tilde{Y}_t^2}{\sigma(t,\tilde{X}_t)\tilde{X}_t}\right\vert^p\mathbbm{1}_{\{\hat{X}_t<\frac{S_0}{2}\}}\right] \le \frac{2^p}{\underline{\sigma}^p}\left(\mathbb{E}^\mathbb{Q}\Big[\max_{0\le s\le 1}\left(\tilde{Y}_s^{4p}\tilde{X}_s^{-2p}\right)\Big]\right)^{\frac{1}{2}}\left(\mathbb{Q}\left[\hat{X}_t<\frac{S_0}{2}\right]\right)^{\frac{1}{2}}\,.\label{proof:u close 2.3}
	\end{align}
Let $Z$ denote a standard normal random variable under the measure $\mathbb{Q}$, and let $N(\cdot)$ denote its cumulative distribution function.
 Since $\hat{X}_t$ is a normal random variable, 
	\begin{equation}\label{eq:hatx vanish}
	\mathbb{Q}\left[\hat{X}_t<\frac{S_0}{2}\right]
	=N\Bigg(-\frac{1}{2\sqrt{\int_0^t\sigma(u,S_0)^2\,du}}\Bigg)
	\le N\left(-\frac{1}{2\overline{\sigma}\sqrt{t}}\right)\,.
	\end{equation}
	It is evident that $N(-\frac{1}{2\overline{\sigma}\sqrt{t}})$ $=o(t^{q})$ as $t\rightarrow{0}$  for any $q>0$. Thus, from Lemma \ref{lem:dummy},
	we obtain \eqref{eq:u app}.

  We now prove the second inequality of \eqref{eq:TF app}. Observe that
	\begin{align}
	&\quad\mathbb{E}^\mathbb{Q}[\vert T\tilde{F}-T\hat{F} \vert^p]\\
    &\le 2^{p-1}\mathbb{E}^\mathbb{Q}\left[\left(\frac{1}{T}\int_0^T \tilde{Y}_t\,dt\right)^{-p}\mathbbm{1}_{\{\frac{1}{T}\int_0^T \hat{Y}_t\,dt<\frac{1}{2}\}}\right] \\
	&+2^{p-1}\mathbb{E}^\mathbb{Q}\left[\left(\frac{1}{T}\int_0^T \vert\tilde{Y}_t-\hat{Y}_t\vert\,dt\right)^{p}\left(\frac{1}{T}\int_0^T \tilde{Y}_t\,dt\right)^{-p}\left(\frac{1}{T}\int_0^T\hat{Y}_t\,dt\right)^{-p}\mathbbm{1}_{\{\frac{1}{T}\int_0^T \hat{Y}_t\,dt\ge \frac{1}{2}\}}\right]. 
	\end{align}
	From Lemmas \ref{lem:x,y close},   \ref{lem:dummy} and H\"older's inequality, for $0\le t \le T$, we have
	\begin{align}
	\mathbb{E}^\mathbb{Q}\left[\left(\frac{1}{T}\int_0^T \vert\tilde{Y}_t-\hat{Y}_t\vert\,dt\right)^{p}\left(\frac{1}{T}\int_0^T \tilde{Y}_t\,dt\right)^{-p}\left(\frac{1}{T}\int_0^T\hat{Y}_t\,dt\right)^{-p}\mathbbm{1}_{\{\frac{1}{T}\int_0^T \hat{Y}_t\,dt\ge \frac{1}{2}\}}\right]\le D_pT^p.
	\end{align}
	By the stochastic Fubini theorem  and Assumption \ref{classical assumption}, it follows that
	\begin{align}
	\mathbb{Q}\left[\frac{1}{T}\int_0^T \hat{Y}_t\,dt<\frac{1}{2}\right]
	&=\mathbb{Q}\left[1+\frac{1}{T}\int_0^T\int_0^T \nu(s,S_0)\mathbbm{1}_{\{s\le t\}}\,dW_s\,dt<\frac{1}{2}\right] \\
	&=\mathbb{Q}\left[\frac{1}{T}\int_0^T \nu(s,S_0)(T-s)\,dW_s<-\frac{1}{2}\right] \\
	&=\mathbb{Q}\left[\frac{1}{T}\sqrt{\int_0^T \nu^2(s,S_0)(T-s)^2\,ds}\,Z<-\frac{1}{2}\right] \\
	&\le N\bigg(-\frac{\sqrt{3}}{2\alpha\sqrt{T}}\bigg)\,.\label{eq:hatY vanish}
	\end{align}
  Since $N(-\frac{\sqrt{3}}{2\alpha\sqrt{T}})=o(T^{q})$ as $T\rightarrow{0}$ for any $q>0$, the second inequality in \eqref{eq:TF app} follows.
  The first inequality in \eqref{eq:TF app} can be shown similarly by applying Lemmas~\ref{lem:x,y close}, \ref{lem:dummy}, and H\"older’s inequality.
\end{proof}


The following lemma shows that the moments of $D_sX_t$, $D_sY_t$, $D_s\tilde{Y}_t$ and $D_s\hat{Y}_t$ are bounded. 

\begin{lemma}\label{lem:Ds bound}
	Let Assumption \ref{classical assumption} hold. Then, for any $p>0$, there exists a positive constant $E_p$ depending only on $p$ such that the following inequalities hold.
	\begin{enumerate}\renewcommand{\labelenumi}{(\roman{enumi})}
			\renewcommand{\theenumi}{\roman{enumi}} 
		\item For $0\le t\le 1$,
		\begin{equation}\label{eq:Dsx bound}
		\underset{s\ge0}{\sup}\,\mathbb{E}^\mathbb{Q}[\vert D_{s}X_t \vert^p] \le E_p.
		\end{equation}
		\item For $0\le t\le 1$,
		\begin{equation}\label{eq:DsY bound}
		\underset{s\ge0}{\sup}\,\mathbb{E}^\mathbb{Q}[\vert D_{s}Y_t \vert^p] \le E_p, \quad
		\underset{s\ge0}{\sup}\,\mathbb{E}^\mathbb{Q}[\vert D_{s}\tilde{Y}_t \vert^p] \le E_p,\quad
		\underset{s\ge0}{\sup}\,\mathbb{E}^\mathbb{Q}[\vert D_{s}\hat{Y}_t \vert^p] \le E_p.
		\end{equation}
		\item For $0\le T\le 1$,
		\begin{equation}\label{eq:TDsF bound}
		\underset{s\ge0}{\sup}\,\mathbb{E}^\mathbb{Q}[\vert TD_{s}F\vert^p] \le E_p,\quad
		\underset{s\ge0}{\sup}\,\mathbb{E}^\mathbb{Q}[\vert TD_{s}\tilde{F}\vert^p] \le E_p,\quad
		\underset{s\ge0}{\sup}\,\mathbb{E}^\mathbb{Q}[\vert TD_{s}^{*}\hat{F}\vert^p] \le E_p.
		\end{equation}
	\end{enumerate}
\end{lemma}

\begin{proof}
	We first provide the proof of \eqref{eq:Dsx bound}. By \cite{benhamou2000application}, the Malliavin derivative $D_sX_t$ is expressed as
	\begin{equation}\label{eq:DsXt}
	D_sX_t=\frac{Y_t}{Y_s}\sigma(s,X_s)X_s\mathbbm{1}_{\{s\le t\}}\,.
	\end{equation}
	Hence, for any $p>0$ and $0\le t\le 1\,$ by Assumption \ref{classical assumption}, Lemma \ref{lem:dummy}, and H\"older's inequality, \eqref{eq:Dsx bound} follows.
	
	We now prove the second and third inequalities in \eqref{eq:DsY bound}. From  \cite[Proposition 1.5.1]{NualartDavid1995TMca}, the Malliavin derivatives $D_s\tilde{Y}_t$ and $D_s\hat{Y}_t$ can be expressed as
	\begin{align}
	D_s\tilde{Y}_t=\nu(s,S_0)\tilde{Y}_t\mathbbm{1}_{\{s\le t\}}\,, \quad D_s\hat{Y}_t=\nu(s,S_0)\mathbbm{1}_{\{s\le t\}}\,.\label{proof:DtildeY, DhatY}
	\end{align}
	Hence, from Assumption \ref{classical assumption}, Lemma \ref{lem:dummy} and H\"older's inequality,  it follows that $\sup_{s\ge 0}\,\mathbb{E}^\mathbb{Q}[\vert D_{s}\tilde{Y}_t \vert^p]$ and $\sup_{s\ge 0}\,\mathbb{E}^\mathbb{Q}[\vert D_{s}\hat{Y}_t \vert^p]$ are bounded by some constant $E_p>0$. To prove the first inequality in \eqref{eq:DsY bound}, we observe that 
	\begin{equation}\label{proof:DY}
	D_sY_t=Y_t\left[\nu(s,X_s)-\int_0^t \nu(u,X_u)\rho(u,X_u)D_sX_u\,du+\int_0^t \rho(u,X_u)D_sX_u\,dW_u\right]\mathbbm{1}_{\{s\le t\}}\,,
	\end{equation}
	which is obtained by  \cite[Proposition 1.5.1]{NualartDavid1995TMca}.
	By Jensen's inequality and the inequalities $\vert\nu\vert\le\alpha$ and $\vert\rho\vert\le\alpha$,
	 we have
	\begin{align}
	\mathbb{E}^\mathbb{Q}[\vert D_sY_t\vert^p]
	&\le 3^{p-1}\alpha^p\mathbb{E}^\mathbb{Q}[Y_t^p] +3^{p-1}\alpha^{2p}\mathbb{E}^\mathbb{Q}\left[Y_t^p\int_0^t \vert D_sX_u\vert^p\,du\right]t^{p-1} \\
	&+3^{p-1}\mathbb{E}^\mathbb{Q}\left[Y_t^p\left\vert\int_0^t \rho(u,X_u)D_sX_u\,dW_u\right\vert^p\right]\,.
	\end{align}
	By Lemma \ref{lem:dummy}, H\"older's inequality
	 and Eq.\eqref{eq:Dsx bound}, there exists a constant $C>0$ such that
	\begin{align}
	\mathbb{E}^\mathbb{Q}\left[Y_t^p\int_0^t \vert D_sX_u\vert^p\,du\right]t^{p-1} \le\left(\mathbb{E}^\mathbb{Q}[\vert Y_t\vert^{2p}]\right)^{\frac{1}{2}}(E_{2p})^{\frac{1}{2}}t^p
	\le C\,\label{proof:DsY bound 1}
	\end{align}
	and
	\begin{align}
	\mathbb{E}^\mathbb{Q}\left[Y_t^p\left\vert\int_0^t \rho(u,X_u)D_sX_u\,dW_u\right\vert^p\right] \le \left(\mathbb{E}^\mathbb{Q}[\vert Y_t\vert^{2p}]\right)^{\frac{1}{2}}(C_p\alpha^{2p}E_{2p}t^p)^{\frac{1}{2}}
	\le C\, \label{proof:DsY bound 3}
	\end{align}
	for all $0\le t\le 1$ and $s\ge0\,.$ This completes the proof of \eqref{eq:DsY bound}.
	
	Finally, we provide the proof of \eqref{eq:TDsF bound}. Similarly to the above, the Malliavin derivative $TD_sF$ can be written as $$TD_sF=-\frac{1}{T}\int_0^T D_sY_t\,dt\left(\frac{1}{T}\int_0^T Y_t\,dt\right)^{-2}\,.$$
 From Lemma \ref{lem:dummy}, H\"older's inequality
  and \eqref{eq:DsY bound}, we obtain \eqref{eq:TDsF bound}.
\end{proof}

\begin{lemma}\label{lem:Ds close}
	Let Assumption \ref{classical assumption} hold. Then, for any $p>0,$ there exists a positive constant $F_p$ depending only on $p$ such that the following inequalities hold.
	\begin{enumerate}\renewcommand{\labelenumi}{(\roman{enumi})}
			\renewcommand{\theenumi}{\roman{enumi}} 
		\item For $0\le t\le 1$,
		\begin{equation}\label{eq:DsY close}
		\underset{s\ge0}{\sup}\,\mathbb{E}^\mathbb{Q}[\vert D_{s}Y_t-D_{s}\tilde{Y}_t \vert^p] \le F_pt^{\frac{p}{2}}, \quad
		\underset{s\ge0}{\sup}\,\mathbb{E}^\mathbb{Q}[\vert D_{s}\tilde{Y}_t-D_{s}\hat{Y}_t \vert^p] \le F_pt^{\frac{p}{2}}.
		\end{equation}
		\item For $0\le T\le 1$,
		\begin{equation}\label{eq:TDsF close}
		\underset{s\ge0}{\sup}\,\mathbb{E}^\mathbb{Q}[\vert TD_{s}F-TD_{s}\tilde{F} \vert^p] \le F_pT^{\frac{p}{2}}, \quad
		\underset{s\ge0}{\sup}\,\mathbb{E}^\mathbb{Q}[\vert TD_{s}\tilde{F}-TD_{s}^{*}\hat{F} \vert^p] \le F_pT^{\frac{p}{2}}.
		\end{equation}
	\end{enumerate}
\end{lemma}

\begin{proof}
It suffices to establish the result for $p \ge 1$. The case $0<p<1$ follows by applying Jensen’s inequality in the same manner as in \eqref{eqn:jensen}. For the first inequality in \eqref{eq:DsY close}, observe that  
	\begin{align}
	&\mathbb{E}^\mathbb{Q}[\vert D_sY_t-D_s\tilde{Y}_t\vert^p]\\
	&\le 3^{p-1}\mathbb{E}^\mathbb{Q}\left[\big\vert Y_t\nu(s,X_s)-\tilde{Y_t}\nu(s,S_0)\big\vert^p\right]\mathbbm{1}_{\{s\le t\}}
	+3^{p-1}\alpha^{2p}\mathbb{E}^\mathbb{Q}\left[Y_t^p\int_0^t \vert D_sX_u\vert^p\,du\right]t^{p-1}\\
	&+3^{p-1}\mathbb{E}^\mathbb{Q}\left[Y_t^p\left\vert\int_0^t \rho(u,X_u)D_sX_u\,dW_u\right\vert^p\right]\,.
	\end{align}
This follows directly from \eqref{proof:DtildeY, DhatY} and \eqref{proof:DY}.
Moreover, by \eqref{proof:DsY bound 1} and \eqref{proof:DsY bound 3}, there exists a constant $F_p>0$ such that $$\mathbb{E}^\mathbb{Q}\left[Y_t^p\int_0^t \vert D_sX_u\vert^p\,du\right]t^{p-1}+\mathbb{E}^\mathbb{Q}\left[Y_t^p\left\vert\int_0^t \rho(u,X_u)D_sX_u\,dW_u\right\vert^p\right]\le F_pt^{\frac{p}{2}}$$ for all $0\le t\le 1$.
Using Lemmas \ref{lem:x,y close} and \ref{lem:dummy}, the inequality
	\begin{align}
	\vert Y_t\nu(s,X_s)-\tilde{Y_t}\nu(s,S_0)\vert\mathbbm{1}_{\{s\le t\}}\le Y_t\vert \nu(s,X_s)-\nu(s,S_0)\vert\mathbbm{1}_{\{s\le t\}}
	+\left\vert \nu(s,S_0)\right\vert \vert Y_t-\tilde{Y_t}\vert\,, 
	\end{align}
 and Theorem 3.4.3 of \cite{zhang2017backward}, we obtain the first inequality in \eqref{eq:DsY close}. The second inequality in \eqref{eq:DsY close} is similarly obtained  from \eqref{proof:DtildeY, DhatY}.

Finally, we  prove \eqref{eq:TDsF close}. From Proposition 1.5.1 in \cite{NualartDavid1995TMca}, it follows that
	\begin{align}
	TD_sF=\frac{-\frac{1}{T}\int_0^TD_sY_t\,dt}{\left(\frac{1}{T}\int_0^TY_t\,dt\right)^2}\,, \quad
	TD_s\tilde{F}=\frac{-\frac{1}{T}\int_0^TD_s\tilde{Y}_t\,dt}{\left(\frac{1}{T}\int_0^T\tilde{Y}_t\,dt\right)^2}\,.
	\end{align}
By Jensen's inequality, we have $\mathbb{E}^\mathbb{Q}[\vert TD_sF-TD_s\tilde{F}\vert^p] \le 2^{p-1}\mathbb{E}^\mathbb{Q}[L_T^p]+2^{p-1}\mathbb{E}^\mathbb{Q}[R_T^p],$
	where $L_T$ and $R_T$ are defined as
	\begin{align}
	L_T
	:=\left\vert \frac{\frac{1}{T}\int_0^TD_sY_t\,dt}{\left(\frac{1}{T}\int_0^TY_t\,dt\right)^2}-\frac{\frac{1}{T}\int_0^TD_sY_t\,dt}{\left(\frac{1}{T}\int_0^T\tilde{Y}_t\,dt\right)^2}\right\vert\,, \quad
	R_T
	:=\left\vert\frac{\frac{1}{T}\int_0^TD_sY_t\,dt}{\left(\frac{1}{T}\int_0^T\tilde{Y}_t\,dt\right)^2}-\frac{\frac{1}{T}\int_0^TD_s\tilde{Y}_t\,dt}{\left(\frac{1}{T}\int_0^T\tilde{Y}_t\,dt\right)^2} \right\vert\,.
	\end{align}
	Using Lemmas \ref{lem:x,y close}, \ref{lem:dummy}, \ref{lem:Ds bound} and H\"older's inequality, we obtain the first inequality in \eqref{eq:TDsF close}.
	
The second inequality in \eqref{eq:TDsF close} can be established in a similar manner, except that we must additionally handle the indicator function $\mathbbm{1}_{\{\frac{1}{T}\int_0^T\hat{Y}_t\,dt\ge\frac{1}{2}\}}$	
appearing in  $D_s^{*}\hat{F}.$ This issue can be addressed using the same technique as in the proof of Lemma~\ref{lem:u,TF app}.
\end{proof}

In Lemma \ref{lem:hat u app}, $\delta(\hat{u})$ is approximated by a normal random variable $\delta(\frac{2}{\sigma(\cdot,S_0)S_0})$.
By applying this lemma, the expectations in \eqref{eq:asian delta mal cal 3} can be directly estimated in terms of multivariate normal random variables.
Consequently, we establish the  asymptotic relations in Proposition~\ref{prop:Asian delta app 3}.

\begin{lemma}\label{lem:hat u app}
	Let Assumption \ref{classical assumption} hold. Then, for any $p>0,$ there exists a positive constant $G_p$ depending only on $p$ such that 
	for all $0\le T\le 1$,
	\begin{equation}\label{eq:hat u app}
	\mathbb{E}^\mathbb{Q}[\vert\delta(\hat{u})\vert^p]\le G_pT^{\frac{p}{2}}\,,\;	\mathbb{E}^\mathbb{Q}\left[\left\vert \delta(\hat{u})-\delta\left(\frac{2}{\sigma(\cdot,S_0)S_0}\right) \right\vert^p\right]\le G_pT^p\,.
	\end{equation}
\end{lemma}

\begin{proof}
	Observe that the Skorokhod integral $\delta(\hat{u})$ coincides with the It\^o integral of $\hat{u}$. By Doob's $L^p$ inequality, the Burkholder--Davis--Gundy inequality, and H\"older's inequality, there is a constant $C_p>0$ such that
	\begin{align}
	\mathbb{E}^\mathbb{Q}[\vert \delta(\hat{u})\vert^p]
	\le C_p\frac{4^p}{\underline{\sigma}^pS_0^p}\left(\frac{2p}{2p-1}\right)^{2p}T^{\frac{p}{2}}
	\end{align} 
	for all $0\le T\le 1.$
	This proves the first inequality in \eqref{eq:hat u app}.
	
We now prove the second inequality in \eqref{eq:hat u app}. For simplicity, we define $g_s:=\hat{u}_s-\frac{2}{\sigma(s,S_0)S_0}$ for $s\in [0,T].$  The process $(g_s)_{0\le s\le T}$ can be written as 
	\begin{align}
	g_s=\left(\hat{u}_s-\frac{2\hat{Y}_s^2}{\sigma(s,S_0)S_0}\mathbbm{1}_{\{\hat{X}_s\ge\frac{S_0}{2}\}}\right)-\frac{2\hat{Y}_s^2}{\sigma(s,S_0)S_0}\mathbbm{1}_{\{\hat{X}_s\le\frac{S_0}{2}\}}+\left(\frac{2\hat{Y}_s^2}{\sigma(s,S_0)S_0}-\frac{2}{\sigma(s,S_0)S_0}\right). 
	\end{align}
Since $\hat{Y}_s$ and $\hat{X}_s$ are normal random variables and $\mathbb{Q}[\hat{X}_s<\frac{S_0}{2}]=o(s^q)$ as $s\rightarrow{0}$  for any $q>0$, from  \cite[Theorem 3.4.3]{zhang2017backward} and H\"older's inequality, we have
	\begin{align}\label{p-th moment of g}
	\mathbb{E}^\mathbb{Q}\left[\left\vert g_s\right\vert ^p \right]\le G_ps^{\frac{p}{2}}.
	\end{align}
	for some positive constant $G_p>0$.
	This completes the proof.
\end{proof}

We recall the processes
$(h_s)_{0\le s\le T}$, $(\hat{h}_s)_{0\le s\le T}$, $(H_s)_{0\le s\le T}$, $(\hat{H}_s)_{0\le s\le T}$ and the random variables $G$, $\hat{G}$ from \eqref{eqn:orig}
and \eqref{eqn:aux}.
Furthermore, we define  
\begin{align}
\tilde{h}_s:=\frac{\tilde{Y}_s}{\sigma(s,\tilde{X}_s)\tilde{X}_s}\,, \; \tilde{H}_s:=\frac{\tilde{X}_T(D_{s}\tilde{Y}_T)}{{\tilde{Y}_T}^2}\,,\;\tilde{G}:=\frac{\tilde{X}_T}{\tilde{Y}_T}\,.
\end{align}
In Lemma \ref{lem:h,G,H app}, the processes $(\tilde{h}_s)_{0\le s\le T}$, $(\hat{h}_s)_{0\le s\le T}$  serve as approximations of $(h_s)_{0\le s\le T}$.
Similarly, $(\tilde{H}_s)_{0\le s\le T}$, $(\hat{H}_s)_{0\le s\le T}$ are employed to approximate $(H_s)_{0\le s\le T}$.
The random variables $\tilde{G}$, $\hat{G}$  are likewise used as approximations of  $G$.
As the proof of the following lemma proceeds analogously to that of Lemma \ref{lem:u,TF app}, it is omitted.
\begin{lemma}\label{lem:h,G,H app}
	Let Assumption \ref{classical assumption} hold. Then, for any $p>0,$ there exists a positive constant $I_p$ depending only on $p$ such that the following inequalities hold. 
	\begin{enumerate}\renewcommand{\labelenumi}{(\roman{enumi})}
			\renewcommand{\theenumi}{\roman{enumi}} 
		\item For  $0\le t\le 1$,
		\begin{equation}
		\mathbb{E}^\mathbb{Q}[\vert h_t-\tilde{h}_t\vert^p] \le I_pt^p\,, \quad
		\mathbb{E}^\mathbb{Q}[\vert \tilde{h}_t-\hat{h}_t\vert^p] \le I_pt^p\,.
		\end{equation}
		\item For $0\le T\le 1$,
		\begin{equation}
		\underset{s\ge0}{\sup}\,\mathbb{E}^\mathbb{Q}[\vert H_s-\tilde{H}_s\vert^p] \le I_pT^{\frac{p}{2}}\,, \quad
		\underset{s\ge0}{\sup}\,\mathbb{E}^\mathbb{Q}[\vert \tilde{H}_s-\hat{H}_s\vert^p] \le I_pT^{\frac{p}{2}}\,.
		\end{equation}
		\item For $0\le T\le 1$,
		\begin{equation}
		\mathbb{E}^\mathbb{Q}[\vert G-\tilde{G}\vert^p] \le I_pT^p\,, \quad
		\mathbb{E}^\mathbb{Q}[\vert \tilde{G}-\hat{G}\vert^p] \le I_pT^p\,.
		\end{equation}
	\end{enumerate}
\end{lemma}

We approximate $\delta(\hat{h})$ by the normal random variable  $\delta(\frac{1}{\sigma(\cdot,S_0)S_0})$ in the following lemma. The proof is similar to that of Lemma \ref{lem:hat u app}; hence, it is omitted.
\begin{lemma}\label{lem:hat h app}
Let Assumption \ref{classical assumption} hold. 
Then, for any $p>0,$ there exists a positive constant $J_p$ depending only on $p$ such that for all $0\le T\le 1$,
	\begin{equation}
	\mathbb{E}^\mathbb{Q}[\vert\delta(\hat{h})\vert^p]\le J_pT^{\frac{p}{2}}\,,\; \mathbb{E}^\mathbb{Q}\left[\left\vert\delta(\hat{h})-\delta\left(\frac{1}{\sigma(\cdot,S_0)S_0}\right)\right\vert^p\right]\le J_pT^p\,.
	\end{equation}
\end{lemma}

\section{The proof of main results in Section \ref{sec:Special case} }
\label{proof:otm delta}

Section \ref{sec:Special case} primarily considers the call and put option payoffs, given by $\Phi(x) = (x - K)_+$ and $\Phi(x) = (K - x)_+$, respectively. The following lemma holds for any Lipschitz function $\Phi.$

\begin{lemma}\label{lem:another repre delta}
	Under  Assumptions \ref{classical assumption} and \ref{Holder payoff assumption} with $\gamma=1,$ 
	 we have 
	\begin{align}
	&\frac{\partial}{\partial S_0}\mathbb{E}^\mathbb{Q}\left[\Phi\left(\frac{1}{T}\int_0^T S_t\,dt\right)\right]
	=\mathbb{E}^\mathbb{Q}\left[\Phi'\left(\frac{1}{T}\int_0^T S_t\,dt\right)\frac{1}{T}\int_0^T Z_t\,dt\right]\,,\\
	&\frac{\partial}{\partial S_0}\mathbb{E}^\mathbb{Q}[\Phi(S_T)]
	=\frac{1}{S_0}\mathbb{E}^\mathbb{Q}[\Phi'(S_T)Z_T]\,.
	\end{align}
	where $Z$ is the unique solution to the SDE \eqref{eq:process z}.
	Here, the derivative $\Phi'$ is defined almost everywhere with respect to the Lebesgue measure.
\end{lemma}
\begin{proof} 
	For any $x>0$, let $S^x$ be the solution to \eqref{eq:S_t} with $S_0=x.$
	From \cite[Theorem 3.4.3]{zhang2017backward},   it follows that
	\begin{align}
\sup_{\vert h\vert \le 1}	\mathbb{E}^\mathbb{Q}\left[\left\vert\frac{1}{h}\left(\Phi\left(\frac{1}{T}\int_0^T S_t^{x+h}\,dt\right)-\Phi\left(\frac{1}{T}\int_0^T S_t^{x}\,dt\right)\right)\right\vert^p\right] <\infty
	\end{align}
for any $p>1$. Thus,  by  \cite[Theorems 6.21 and 6.25 ]{klenke2013probability} and  \cite[Lemma 5.2.3]{zhang2017backward}, we obtain the desired results.
\end{proof}

The proof of Theorem \ref{thm:otm delta} is presented below.
The rate function $\mathcal{I}$ in \eqref{eq:rate function} satisfies the following property. For any Borel set $A$ in $\mathbb{R}^{+},$
	\begin{align}\label{eqn:LDP}
	-\inf_{x\in A^{\circ}}\mathcal{I}(x,S_0)
	&\le \liminf_{T\rightarrow{0}}T\log\left(\mathbb{Q}\left[\frac{1}{T}\int_0^T S_t\,dt\in A\right]\right) \\
	&\le \limsup_{T\rightarrow{0}}T\log\left(\mathbb{Q}\left[\frac{1}{T}\int_0^T S_t\,dt\in A\right]\right)
	\le -\inf_{x\in \overline{A}}\mathcal{I}(x,S_0)\,,
	\end{align}
	where $A^{\circ}$ is the interior of $A$ and $\overline{A}$ is the closure of $A.$ See \cite{DemboAmir1998Ldta, PirjolDan2016SMAO} for details.

\begin{proof}
	We prove only \eqref{eq:otm cdelta}. Let $Z$ be the  solution to \eqref{eq:process z}. From Lemma \ref{lem:another repre delta} and H\"older's inequality, we have
	\begin{align}
	\Delta_A^{\textnormal{call}}(T)
	&=\frac{e^{-rT}}{S_0}\mathbb{E}^\mathbb{Q}\left[\frac{1}{T}\int_{0}^{T} Z_t\,dt\,\mathbbm{1}_{\{\frac{1}{T}\int_{0}^{T} S_t\,dt \ge K\}}\right] \\
	&\le\frac{e^{-rT}}{S_0}\left(\frac{1}{T}\int_{0}^{T}S_0^pe^{p(r-q)t}e^{K(p)\overline{\sigma}^2t}\,dt\right)^{\frac{1}{p}}\left(\mathbb{Q}\left[\frac{1}{T}\int_{0}^{T}S_t\,dt \ge K\right]\right)^{\frac{1}{p'}}
	\end{align}
	for $p,p'>0$ with $1/p+1/p'=1$.
	Letting $T\rightarrow{0},$ we obtain the  inequality  
	\begin{equation}
	\limsup_{T\rightarrow0}T\log\Delta_A^{\textnormal{call}}(T)\le\frac{-\mathcal{I}(K,S_0)}{p'}
	\end{equation}
	from \eqref{eqn:LDP}.
By taking the limit	  $p'\rightarrow{1}$, this yields the upper bound
$$	\limsup_{T\rightarrow0}T\log\Delta_A^{\textnormal{call}}(T)\le -\mathcal{I}(K,S_0)\,.$$
	The lower bound $$-\mathcal{I}(K,S_0)\le 	\limsup_{T\rightarrow0}T\log\Delta_A^{\textnormal{call}}(T)$$follows from \eqref{eqn:LDP}, Lemma \ref{lem:dummy general}, the reverse H\"older inequality, and the positivity of 
	$Z$. This completes the proof.
\end{proof}

The proof of Corollary
 \ref{cor:itm delta} is presented below.
\begin{proof}
	The put-call parity for the Asian option delta yields
	\begin{equation}
	\Delta_A^{\textnormal{call}}(T)-\Delta_A^{\textnormal{put}}(T)
	=\frac{e^{-rT}}{S_0}\frac{1}{T}\int_{0}^{T} \mathbb{E}^\mathbb{Q}[S_t]\,dt
	=\frac{e^{-rT}}{S_0}\frac{1}{T}\int_{0}^{T} S_0e^{(r-q)t}\,dt
	=\frac{e^{-qT}-e^{-rT}}{(r-q)T}\,.
	\end{equation}
From Theorem \ref{thm:otm delta}, the out-of-the-money Asian delta decays at an exponential rate.
Hence, the Taylor expansion
	\begin{equation}
	\frac{e^{-qT}-e^{-rT}}{(r-q)T}=1-\frac{1}{2}(r+q)T+\left(\frac{r^2+rq+q^2}{6}\right)T^2+\mathcal{O}(T^3)\,
	\end{equation}
implies Corollary \ref{cor:itm delta}.
\end{proof}

\end{document}